\documentclass[journal]{IEEEtran}
\usepackage{amssymb,amsthm}
\usepackage{amsmath,amsfonts}
\usepackage{array}
\usepackage{algorithmic}
\usepackage{algorithm}
\usepackage{array}
\usepackage[caption=false,font=normalsize,labelfont=sf,textfont=sf]{subfig}
\usepackage{cite}
\usepackage{color}
\usepackage{graphicx}
\usepackage{textcomp}
\usepackage{stfloats}
\usepackage{url}
\usepackage{verbatim}
\usepackage{makecell}
\usepackage{multirow}
\usepackage{harpoon}

\newtheorem{theorem}{Theorem}
\newtheorem{lemma}{Lemma}
\theoremstyle{definition}

\theoremstyle{plain}

\newcommand{\lmref}[1]{Lemma \ref{#1}}
\newcommand{\thref}[1]{Theorem \ref{#1}}

\newcommand{\figref}[1]{Fig.~\ref{#1}}
\newcommand{\tabref}[1]{Table \ref{#1}}
\newcommand{\alref}[1]{Algorithm \ref{#1}}
\newcommand{\appref}[1]{Appendix \ref{#1}}
\newcommand{\secref}[1]{Section \ref{#1}}

\newcommand{\diag}[1]{\mathrm{diag}\left(#1\right)}
\newcommand{\logdet}[1]{\log\det\left(#1\right)}

\newcommand{\logtwo}[1]{\log_{2}\left(#1\right)}

\newcommand{\argmin}[1]{\mathop{\arg\min}\limits_{#1}}

\newcommand{\cA}{\mathcal{A}}

\newcommand{\cE}{\mathcal{E}}
\newcommand{\cF}{\mathcal{F}}

\newcommand{\cL}{\mathcal{L}}
\newcommand{\cM}{\mathcal{M}}
\newcommand{\cN}{\mathcal{N}}

\newcommand{\cR}{\mathcal{R}}

\newcommand{\cT}{\mathcal{T}}
\newcommand{\cU}{\mathcal{U}}

\newcommand{\bb}{\mathbf{b}}

\newcommand{\be}{\mathbf{e}}

\newcommand{\bw}{\mathbf{w}}
\newcommand{\bx}{\mathbf{x}}
\newcommand{\by}{\mathbf{y}}
\newcommand{\bz}{\mathbf{z}}

\newcommand{\bA}{\mathbf{A}}
\newcommand{\bB}{\mathbf{B}}
\newcommand{\bC}{\mathbf{C}}
\newcommand{\bD}{\mathbf{D}}
\newcommand{\bE}{\mathbf{E}}
\newcommand{\bF}{\mathbf{F}}
\newcommand{\bG}{\mathbf{G}}
\newcommand{\bH}{\mathbf{H}}
\newcommand{\bI}{\mathbf{I}}

\newcommand{\bM}{\mathbf{M}}

\newcommand{\bP}{\mathbf{P}}
\newcommand{\bQ}{\mathbf{Q}}
\newcommand{\bR}{\mathbf{R}}

\newcommand{\bU}{\mathbf{U}}
\newcommand{\bV}{\mathbf{V}}

\newcommand{\bX}{\mathbf{X}}
\newcommand{\bY}{\mathbf{Y}}
\newcommand{\bZ}{\mathbf{Z}}

\newcommand{\bbC}{\mathbb{C}}
\newcommand{\bbR}{\mathbb{R}}

\newcommand{\bzero}{\mathbf{0}}

\newcommand{\bLambda}{{\boldsymbol\Lambda}}

\newcommand{\bPi}{{\boldsymbol\Pi}}

\newcommand{\bGamma}{{\boldsymbol\Gamma}}

\newcommand{\bxi}{{\boldsymbol\xi}}

\newcommand{\bzeta}{{\boldsymbol\zeta}}

\def\BibTeX{{\rm B\kern-.05em{\sc i\kern-.025em b}\kern-.08em
    T\kern-.1667em\lower.7ex\hbox{E}\kern-.125emX}}
\usepackage{balance}
\begin{document}
\title{ Precoder Design for Massive MIMO Downlink with Matrix Manifold Optimization }
\author{Rui~Sun,~\IEEEmembership{Student Member, IEEE}, Chen~Wang,~\IEEEmembership{Student Member, IEEE}, An-An~Lu,~\IEEEmembership{Member, IEEE}, Xiqi~Gao,~\IEEEmembership{Fellow, IEEE}, and Xiang-Gen~Xia,~\IEEEmembership{Fellow, IEEE}
	
	\thanks{This work was supported by the National Key R\&D Program of China under Grant 2018YFB1801103, the Jiangsu Province Basic Research Project under Grant BK20192002, the Fundamental Research Funds for the Central Universities under Grant 2242022k60007, the Key R\&D Plan of Jiangsu Province under Grant BE2022067, the Huawei Cooperation Project, and the National Natural Science Foundation of China under Grants  62394294 and 62371125.	  
		\textit{(Corresponding author: Xiqi Gao.)} }
	
\thanks{Rui Sun, Chen Wang, An-An Lu and Xiqi Gao are with the National Mobile Communications Research Laboratory,
	Southeast University, Nanjing 210096, China and are also with Purple Mountain Laboratories, Nanjing 211111, China (e-mail:
	ruisun@seu.edu.cn; wc@seu.edu.cn; aalu@seu.edu.cn;  xqgao@seu.edu.cn).}
\thanks{Xiang-Gen Xia is with the Department of Electrical and Computer Engineering, University of Delaware, Newark, DE 19716
	USA (e-mail: xxia@ee.udel.edu).}
}


\maketitle

\begin{abstract}
We investigate the weighted sum-rate (WSR) maximization linear precoder design for massive multiple-input multiple-output (MIMO) downlink. We consider a single-cell system with multiple users and propose a unified matrix manifold optimization framework applicable to total power constraint (TPC), per-user power constraint (PUPC) and per-antenna power constraint (PAPC). We prove that the precoders under TPC, PUPC and PAPC are on distinct Riemannian submanifolds, and transform the constrained problems in Euclidean space to unconstrained ones on manifolds. In accordance with this, we derive Riemannian ingredients, including orthogonal projection, Riemannian gradient, Riemannian Hessian, retraction and vector transport, which are needed for precoder design in the matrix manifold framework. Then, Riemannian design methods using Riemannian steepest descent, Riemannian conjugate gradient and Riemannian trust region are provided to design the WSR-maximization precoders under TPC, PUPC or PAPC. Riemannian methods do not involve the inverses of the large dimensional matrices during the iterations, reducing the computational complexities of the algorithms.  Complexity analyses  and performance simulations demonstrate the advantages of the proposed precoder design. 
\end{abstract}


\section{Introduction}\label{Intro}
\IEEEPARstart{M}{assive} multiple-input multiple-output (MIMO) is one of the key techniques in the fifth generation (5G) wireless networks and will play an important role in future 6G systems with further increased antenna number scale \cite{bjornson_massive_2019,carvalho_non-stationarities_2020}. In massive MIMO systems, the base station (BS) equipped with a large number of antennas can serve a number of user terminals (UTs) on the same time-frequency resource, providing enormous potential capacity gains and high energy efficiency  \cite{Fundamentals,MMIMO}. However, serving many users simultaneously causes serious inter-user interference, which may reduce the throughput. To suppress the interference and increase the system throughput, precoders should be properly designed for massive MIMO downlink (DL) transmission. Specifically, a linear precoder is particularly attractive due to its low complexity \cite{LinearPrecoding1,LinearPrecoding2,LinearPrecoding3}.

There are several criteria for the DL precoder design, including minimum mean square error (MMSE), weighted sum-rate (WSR), quality of service (QoS), signal-to-interference-plus-noise ratio (SINR), signal-to-leakage-plus-noise ratio (SLNR), energy efficiency (EE), etc. \cite{MMSE,WMMSE,QoS,PUPC2,SLNR,EE}. Among these criteria, WSR is of great practical significance and widely considered in massive MIMO systems, as it directly aims at increasing the system throughput, which is one of the main objectives of communication systems. In addition, different power constraints may be imposed on the designs. Generally, power constraints for massive MIMO systems can be classified into three main categories: total power constraint (TPC), per-user power constraint (PUPC) and per-antenna power constraint (PAPC).
Designing precoders under different power constraints usually forms different problems, for which different approaches are proposed in the literature. To be specific, the precoder design problem under TPC is typically formulated as a nonconvex WSR-maximization problem, which is transformed into an MMSE problem iteratively solved by alternating optimization in \cite{WMMSE}. Alternatively, \cite{MM} solves the  WSR-maximization problem in the framework of the minorize-maximize (MM) method by finding a surrogate function of the objective function.   When designing precoders under PUPC, the received SINR of each user is usually investigated with QoS constraints. With the received SINR, the sum-rate maximization problem is formulated and solved by the Lagrange dual function \cite{PUPC1, PUPC2}. Precoders under PAPC are deeply coupled with each other and thus challenging to design. In  \cite{ZFBased}, precoders under PAPC are designed by  minimizing the distance between the precoders under TPC and those under PAPC.  The WSR-maximization precoder design under PAPC using the dual coordinate ascent method (DCAM)  is proposed in \cite{DCAM}, where  the original nonconvex
WSR-maximization problem is approximated as a sequence of convex
MMSE-minimization problems subject to PAPC,  and each convex problem
is solved by the DCAM.

%

In general, the designs of WSR-maximization precoders under the power constraints mentioned above can be formulated as optimization problems with equality constraints. Recently, manifold optimization has been extensively studied and successfully applied to many domains \cite{Manifold1,Manifold2,Manifold4,Manifold5}, showing a great advantage in dealing with smooth objective functions with  challenging equality constraints.  In mathematics, a manifold is a topological space that locally resembles Euclidean space near each point. By revealing the inherent geometric properties of the equality constraints, manifold optimization reformulates the constrained problems  in Euclidean space as unconstrained ones on manifold. By defining the Riemannian ingredients associated with a \textit{Riemannian manifold}, several Riemannian methods are presented for solving the unconstrained problems on manifold. In addition, manifold optimization usually shows promising algorithms resulting from the combination of insight from differential geometry, optimization, and numerical analysis.  To be specific, by revealing that the precoders under TPC, PUPC or PAPC are on different Riemannian submanifolds, we can leverage this insight to transform the constrained problems into unconstrained ones on these submanifolds. Therefore, manifold optimization can provide a potential way for designing optimal WSR-maximization precoders under different power constraints in a unified framework. 

In this paper, we focus on WSR-maximization precoder design for massive MIMO DL transmission and propose  a matrix manifold framework applicable to TPC, PUPC and PAPC. We reveal the geometric properties of the precoders under different power constraints and prove that the precoder sets satisfying TPC, PUPC and PAPC form three different Riemannian submanifolds, respectively, transforming the constrained problems in Euclidean space into unconstrained ones on Riemannian submanifolds. To facilitate a better understanding, we analyze the precoder designs under TPC, PUPC and PAPC in detail. All the ingredients required during the optimizations on Riemannian submanifolds are derived for the three power constraints. Further, we present three Riemannian design methods using Riemannian steepest descent (RSD), Riemannian conjugate gradient (RCG) and Riemannian trust region (RTR), respectively. Without the need to invert the large dimensional matrix during the iterations, Riemannian methods can efficiently save computational costs, which is beneficial in practice. Complexity analysis shows that the method using RCG is computationally efficient. The numerical results confirm the  advantages of the RCG method in convergence speed and WSR performance.

The remainder of the paper is organized as follows. In \secref{Framework}, we introduce the preliminaries in the matrix manifold optimization. In \secref{Precoding}, we first formulate the WSR-maximization precoder design problem in Euclidean space. Then, we transform the constrained problems in Euclidean space under TPC, PUPC and PAPC to the unconstrained ones on Riemannian submanifolds and derive Riemannian ingredients in the matrix manifold framework. To solve the unconstrained problems on the Riemannian submanifolds,  \secref{RiemannianAlgorithms} provides three Riemannian design methods and their complexity analyses. \secref{Results} presents numerical results and discusses the performance of the proposed precoder designs. The conclusion is drawn in \secref{Conclusion}.

\textit{Notations:} Boldface lowercase and uppercase letters represent the column vectors and matrices, respectively. We write conjugate transpose of matrix $ \bA $ as   $ \bA^{H} $ while $ \mathrm{tr}\left(\bA\right) $ and $ \det\left( \bA\right)  $ denote the matrix trace and determinant of $ \bA $, respectively. $ \Re \left\lbrace \bA \right\rbrace $ means the real part of $ \bA $.  Let the mathematical expectation be $ \mathbb{E}\left\{\cdot\right\}$. $ \bI_{M} $ denotes the $ M \times M $ dimensional identity matrix, whose subscript may be omitted for brevity. $ \bzero $ represents the vector or matrix whose elements are all zero. The Hadamard product of $ \bA $ and $ \bB $ is $ \bA \odot \bB $. Let $ \be_j  $ denote the vector with the $j$-th element equals $1$ while the others equal $0$. $ \diag{\bb} $ represents the diagonal matrix with $ \bb $ along its main diagonal and $ \diag{\bA} $ denotes the column vector of the main diagonal of $ \bA $. Similarly, $\bD=\mathrm{blkdiag}\left\{\bA_1,\cdots,\bA_K\right\}$ denotes the block diagonal matrix with $\bA_1,\cdots,\bA_K$ on the diagonal and $\left[\bD\right]_i$ denotes the $i$-th matrix on the diagonal, i.e., $\bA_i$. A mapping $ F $ from manifold $\cN$ to manifold $\cM$ is $F: \cN\rightarrow \cM: \bX \mapsto \bY$ denoted as $F(\bX)=\bY$.  The differential of $ F\left( \bX\right) $ is represented as $ \mathrm{D}F\left( \bX\right) $ while $ \mathrm{D}F\left( \bX\right) \left[ \bxi_{\bX} \right]  $ or $ \mathrm{D}F\left[ \bxi_{\bX} \right]  $ means the directional derivative of $ F $ at $ \bX $ along the tangent vector $ \bxi_{\bX} $. 
\section{Preliminaries for Matrix Manifold Optimization}\label{Framework}
In this section, we introduce some basic notions in the matrix manifold optimization  and list some common notations  in \tabref{Notation}, with $\cM$ representing a Riemannian manifold. For notational simplicity, the superscripts for distinguishing different manifolds can be neglected when no confusion will arise. To avoid overwhelmingly redundant preliminaries, we focus on a coordinate-free analysis and omit charts and differential structures. For a complete introduction to smooth manifolds, we refer to references \cite{Absil2009,Boumal2020, Lee}.
\renewcommand\arraystretch{1.2}
\begin{table*}
	\centering
	\caption{Notations of the ingredients on a Riemannian submanifold $\cM$}
	\label{Notation}
	\begin{tabular}{|l|l|l|l|}
		\hline
		$\bX$ &A point on $\cM$ &$T_{\bX}\cM$ &Tangent space of $\cM$ at $\bX$\\
		\hline
		$\bxi_{\bX}$ & Tangent vector in $T_{\bX}\cM$ &$N_{\bX}\cM$&Normal space of $\cM$ at $\bX$ \\ 
		\hline
		$f(\bX)$&Smooth function defined on $\cM$  &$\mathrm{D}f\left(\bX\right)\left[\cdot\right]$ &Differential operator of $f\left(\bX\right)$ \\
		\hline
		$g_{\bX}^{\cM}\left(\cdot\right)$&Riemannian metric operator on $T_{\bX}\cM$ &$R_{\bX}^{\cM}\left(\cdot\right)$& Retraction operator from $T_{\bX}\cM$  to $\cM$\\
		\hline
		$\mathrm{grad}f(\bX)$&Riemannian gradient of $f(\bX)$  &$\mathrm{Hess}f(\bX)\left[\cdot\right]$&Riemannian Hessian operator of $f(\bX)$\\
		\hline
		$\bPi_{T_{\bX}{\cM}}^{T_{\bY}{\cN}} \left(\bxi_{\bX}\right) $&Orthogonal projection from $T_{\bY}{\cN}$ to $T_{\bX}{\cM}$&$\cT_{\eta_{\bX}}^{\cM}\left(\cdot\right)$ &Vector transport operator to $T_{R_{\bX}^{\cM}\left(\eta_{\bX}\right)}\cM$\\
		\hline
	\end{tabular}
\end{table*}

For a manifold $\cN$, a smooth mapping $\gamma: \mathbb{R} \rightarrow \cN: t \mapsto \gamma(t)$ is termed a curve in $\cN$. Let $\bX$ be a point on $\cN$ and $\cF_{\bX}(\cN)$ denote the set of smooth real-valued functions defined on a neighborhood of $\bX$. A \textit{tangent vector} $\bxi_{\bX}$ to a manifold $\cN$ at a point $\bX$ is a mapping from $\cF_{\bX}(\cN)$ to $\mathbb{R}$ such that there exists a curve $\gamma$ on $\cN$ with $\gamma(0)=\bX$, satisfying $\bxi_{\bX} f=\left.\frac{\mathrm{d}(f(\gamma(t)))}{\mathrm{d} t}\right|_{t=0}$ for all $f \in \cF_{\bX}(\cN)$. Such a curve $\gamma$ is said to realize the tangent vector $\bxi_{\bX}$.  For a manifold $\cN$, every point $\bX$ on the manifold is attached to a unique and linear \textit{tangent space}, denoted by $T_{\bX}\cN$. $T_{\bX}\cN$ is a set of all the tangent vectors to $\cN$ at $\bX$. The \textit{tangent bundle} is the set of all tangent vectors to $\cN$ denoted by $T\cN$, which is also a manifold. A \textit{vector field} $\bxi$ on a manifold $\cN$ is a smooth mapping from $\cN$ to $T\cN$ that assigns to each point $\bX\in\cN$ a tangent vector $\bxi_{\bX}\in T_{\bX}\cN$. Particularly, every vector space $\cE$ forms a \textit{linear manifold} naturally. Assuming $\cN$ is a vector space $\cE$, we have $T_{\bX}\cN=\cE$. 

\begin{figure}
	\centering                 				
	\includegraphics[scale=0.37]{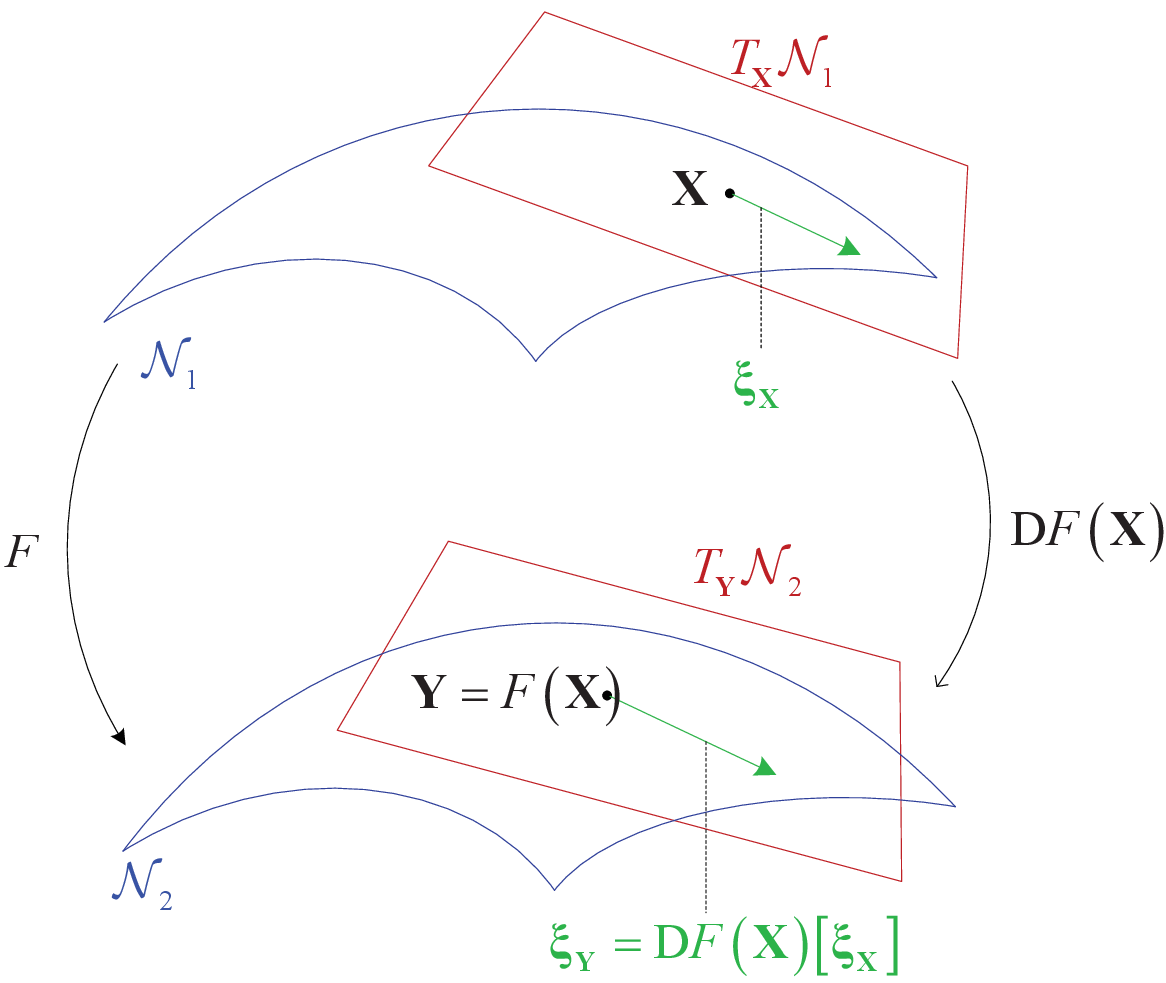}
	\caption{The relationship between two manifolds \cite{Absil2009}. $F$ is a smooth mapping from $\cN_1$  to $\cN_2$, and $\mathrm{D}F\left(\bX\right)$ is a linear mapping from $T_{\bX}{\cN_1}$ to $T_{\bY}{\cN_2}$.}
	\label{fig_manifold}
\end{figure}

Let $F:\cN_1 \rightarrow \cN_2$ be a smooth mapping from manifold $\cN_1$ to manifold $\cN_2$. The differential of $F$ at $\bX$ is a linear mapping from $T_{\bX}\cN_1$ to $T_{\bY}\cN_2$, denoted by $\mathrm{D}F\left(\bX\right)\left[\cdot\right]$. Denote $\bY=F\left(\bX\right)$ as a point on $\cN_2$. Similarly, $\bxi_{\bY}=\mathrm{D}F\left(\bX\right)\left[\bxi_{\bX}\right]$ is a tangent vector to $\cN_2$ at $\bY$, i.e.,  an element in $T_{\bY}\cN_2$. In particular, if $\cN_1$ and $\cN_2$ are linear manifolds, $\mathrm{D}F\left(\bX\right)$ reduces to the classical directional derivative \cite{ManifoldDifferential}
\begin{equation}
	\begin{aligned}
 		\mathrm{D}F\left(\bX\right)\left[ \boldsymbol{\bxi}_{\bX} \right] = \lim_{t\rightarrow 0}\frac{F\left(\bX+t\boldsymbol{\bxi}_{\bX}\right)-F\left(\bX\right)}{t}.
 	\end{aligned}
\end{equation}
The rank of $F$ at $\bX$ is the dimension of the range of $\mathrm{D}F\left(\bX\right)\left[ \boldsymbol{\bxi}_{\bX} \right]$. \figref{fig_manifold} is a simple illustration for geometric understanding.

In practice, the equality constraint $F:\cN \rightarrow 0$ in many optimization problems forms a mapping between two linear manifolds. The solutions satisfying the equality constraint constitute a set $\cM=\left\{ \bX \in \cN  \mid F(\bX)=0 \right\}$. The set $\cM$ may admit several manifold structures, while it admits at most one differential structure that makes it an \textit{embedded submanifold}  of $\cN$.  Whether $\cM$ forms an embedded submanifold mainly depends on the properties of $F\left(\bX\right)$ \cite[Proposition 3.3.2]{Absil2009}. We assume $\cM$ is an embedded submanifold of $\cN$ in the rest of this section to facilitate the introduction. Note that $\bX \in \cM$ is still a point on $\cN$.
A level set of a real-valued function $F$ is the set of values $\bX$ for which $F\left({\bX}\right)$ is equal to a given constant. In particular, when $\cM$ is defined as a level set of a constant-rank function $F$, the tangent space of $\cM$ at $\bX$ is the kernel of the differential of $F$ and a subspace of the tangent space of $\cN$:  
\begin{equation}\label{ker}
 	\begin{aligned}
 		T_{\bX}{\cM}=\mathrm{ker}\left(\mathrm{D}F\left(\bX\right)\right)\subseteq T_{\bX}\cN.
 	\end{aligned}
\end{equation}

By endowing tangent space  with an inner product $g_{\bX}^{\cN}\left(\cdot\right)$, we can define the length of tangent vectors in $T_{\bX}\cN$. Note that the subscript and the superscript  in $g_{\bX}^{\cN}\left(\cdot\right)$ is used to distinguish the metric of different points on different manifolds for clarity, which may be omitted if no confusion will arise. $g_{\bX}^{\cN}\left(\cdot\right)$ is called \textit{Riemannian metric} if it varies smoothly and the manifold is called Riemannian manifold. Since $T_{\bX}{\cM}$ can be regarded as a subspace of $T_{\bX}{\cN}$, the Riemannian metric $g_{\bX}^{\cN}\left(\cdot\right)$ on $\cN$ induces a Riemannian metric $g_{\bX}^{\cM}\left(\cdot\right)$ on $\cM$ according to
\begin{equation}\label{RiemannianMetricSubmanifold}
	\begin{aligned}
		g_{\bX}^{\cM}\left(\bxi_{\bX},\bzeta_{\bX}\right)=g_{\bX}^{\cN}\left(\bxi_{\bX},\bzeta_{\bX}\right), \forall \bxi_{\bX},\bzeta_{\bX}\in T_{\bX}\cM.
	\end{aligned}
\end{equation}
Note that $\bxi_{\bX}$ and $\bzeta_{\bX}$ on the right hand side are viewed as elements in $T_{\bX}\cN$.  Endowed with the Riemannian metric \eqref{RiemannianMetricSubmanifold}, the embedded submanifold $\cM$ forms a \textit{Riemannian submanifold}.

In practice, it is suggested to utilize the \textit{orthogonal projection} to obtain the elements in $T_{\bX}{\cM}$. 
With the Riemannian metric $g_{\bX}^{\cN}\left(\cdot\right)$, the $T_{\bX}{\cN}$ can be divided into two orthogonal subspaces as
\begin{equation}\label{RawComposition}
	\begin{aligned}
		T_{\bX}{\cN}=T_{\bX}{\cM}\oplus N_{\bX}{\cM},
	\end{aligned}
\end{equation}
where the normal space $N_{\bX}{\cM}$ is the orthogonal complement of $T_{\bX}{\cM}$ defined as
\begin{equation}\label{Normal_Space}
	\begin{aligned}
		N_{\bX}{\cM}=\left\{ \bxi_{\bX}\in T_{\bX}\cN \mid g_{\bX}\left(\bxi_{\bX},\bzeta_{\bX}\right)=0,\ \forall \bzeta_{\bX}\in T_{\bX}{\cM}  \right\}.
	\end{aligned}
\end{equation}
Thus, any $\bxi_{\bX}\in T_{\bX}{\cN} $ can be uniquely decomposed into the sum of an element in $T_{\bX}{\cM}$ and an element in $N_{\bX}{\cM}$
\begin{equation}\label{Projection_DirectSum}
	\begin{aligned}
		\bxi_{\bX}=\bPi_{T_{\bX}{\cM}}^{T_{\bX}{\cN}} \left(\bxi_{\bX}\right)+\bPi_{N_{\bX}{\cM}}^{T_{\bX}{\cN}}\left(\bxi_{\bX}\right),
	\end{aligned}
\end{equation}
 where $\bPi_{\bB}^{\bA}\left(\cdot\right)$ denotes the orthogonal projection from $\bA$ to $\bB$.

For a smooth real-valued function $f$ on a Riemannian submanifold $\cM$, the \textit{Riemannian gradient} of $f$ at $\bX$, denoted by $\mathrm{grad}f_{\cM}(\bX)$, is defined as the unique element in $T_{\bX}{\cM}$ that satisfies
 \begin{equation}\label{RawGradient}
 	\begin{aligned}
 		g_{\bX}^{\cM}\left( \mathrm{grad}f_{\cM}(\bX),\boldsymbol{\bxi}_{\bX} \right) = \mathrm{D}f\left(\bX\right)\left[ \boldsymbol{\bxi}_{\bX} \right], \forall \bxi_{\bX}\in T_{\bX}\cM.
 	\end{aligned}
 \end{equation}
Denote $\mathrm{grad}f_{\cM}$ as the vector field of the Riemannian gradient on $\cM$, whose subscript can be omitted when no confusion will arise. Note that, since $\mathrm{grad}f_{\cM}(\bX)\in T_{\bX}{\cM}$, we have $\mathrm{grad}f_{\cN}(\bX)\in T_{\bX}\cN$ and it can then be decomposed as \eqref{Projection_DirectSum}.

When second-order derivative optimization algorithms, such as Newton method and trust-region method, are preferred, \textit{affine connection} is indispensable \cite{Absil2009}. Let $\mathfrak{X}(\cN)$ denote the set of smooth vector fields on $\cN$, the affine connection on $\cN$ is defined as a mapping $\nabla^{\cN}: \mathfrak{X}(\cN) \times \mathfrak{X}(\cN) \rightarrow \mathfrak{X}(\cN):(\boldsymbol{\eta}, \boldsymbol{\bxi}) \mapsto \nabla_{\boldsymbol{\eta}}^{\cN} \boldsymbol{\bxi}$. When Riemannian manifold $\cM$ is an embedded submanifold of a vector space, $\nabla_{\boldsymbol{\eta}_{\bX}}^{\cM} \bxi$ is called \textit{Riemannian connection} and given by \cite[Proposition 5.3.1]{Absil2009}
\begin{equation}
	\begin{aligned}
		\nabla_{\boldsymbol{\eta}_{\bX}}^{\cM} \bxi=&\bPi_{T_{\bX}{\cM}}^{T_{\bX}{\cN}}\left(	\nabla_{\boldsymbol{\eta}_{{\bX}}}^{\cN} \bxi\right)\\
		=&\bPi_{T_{\bX}{\cM}}^{T_{\bX}{\cN}} \left( \mathrm{D}\boldsymbol{\bxi}\left[\boldsymbol{\eta}_{{\bX}}\right] \right), \forall \boldsymbol{\eta}_{\bX}\in T_{\bX}\cM, \bxi\in \mathfrak{X}(\cM).
	\end{aligned}
\end{equation}
Given a Riemannian connection $\nabla^{\cM}$ defined on $\cM$, the Riemannian Hessian of a real valued function $f$ at point $\bX$ on $\cM$ is the linear mapping $\operatorname{Hess} f(\mathbf{X})\left[\cdot\right]$ of $T_{\bX} \cM$ into itself defined as
\begin{equation}\label{RiemannianConnection}
	\begin{aligned}
		\mathrm{Hess}f(\bX)\left[\boldsymbol{\bxi}_{\bX}\right] = \nabla_{\boldsymbol{\bxi}_{\bX}}^{\cM}\mathrm{grad}f.
	\end{aligned}
\end{equation}

The notion of moving along the tangent vector while remaining on the manifold is generalized by \textit{retraction}. The retraction $R_{\bX}^{\cN}\left(\cdot\right)$, a smooth mapping from $T_{\bX} \cN$ to $\cN$ \cite[Definition 4.1.1]{Absil2009}, builds a bridge between the tangent space $T_{\bX} \cN$ and the manifold $\cN$. Practically, the Riemannian exponential mapping forms a retraction but is computationally expensive to implement. For a Riemannian submanifold $\cM$, the geometric structure of the manifold is utilized to define the retraction $R_{\bX}^{\cM}\left(\cdot\right)$ through the projection for affordable and efficient access. 
\begin{figure}[t]
	\centering                 				
	\includegraphics[scale=0.6]{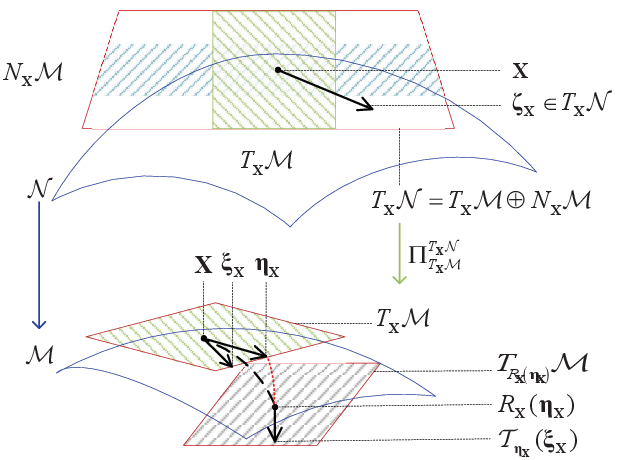}
	\caption{ Geometric interpretation of orthogonal projection, retraction and vector transport.}
	\label{fig_manifold_tangent}
\end{figure}

Sometimes, we need to move the tangent vector from the current tangent space to another, which is not always straightforward as the tangent spaces are different in a nonlinear manifold.  To address this difficulty, \textit{vector transport} denoted by $\cT_{\boldsymbol{\eta}_{\bX}}^{\cN}\left(\bxi_{\bX}\right)\in T_{R_{\bX}\left(\boldsymbol{\eta}_{\bX}\right)}\cN$  is introduced, which specifies how to transport a tangent vector $\bxi_{\bX}$ from a point $\bX$ on $\cN$ to another point $R_{\bX}\left(\boldsymbol{\eta}_{\bX}\right)$. Like $\mathrm{Hess}f(\bX)\left[\cdot\right]$, $\cT_{\boldsymbol{\eta}_{\bX}}^{\cN}\left(\cdot\right)$ is an operator rather than a matrix. For the Riemannian submanifold, $\cT_{\boldsymbol{\eta}_{\bX}}^{\cM}\left(\bxi_{\bX}\right)$ can be achieved by orthogonal projection \cite{Absil2009}:
\begin{equation}\label{RawVectorTransport}
	\begin{aligned}
		\cT_{\boldsymbol{\eta}_{\bX}}^{\cM}\left(\bxi_{\bX}\right)=\Pi_{T_{R_{\bX}^{\cM}\left(\boldsymbol{\eta}_{\bX}\right)}\cM}^{T_{\bX}\cM}\left( \bxi_{\bX}\right).
	\end{aligned}
\end{equation}
For geometric understanding, \figref{fig_manifold_tangent} is a simple illustration.

\section{Problem Formulation and Riemannian Elements in Matrix Manifold framework }\label{Precoding}
In this section,  we first present the WSR-maximization precoder design problems under TPC, PUPC and PAPC, respectively, for massive MIMO DL in Euclidean space. Then, we show that the precoder set satisfying TPC forms a \textit{sphere} and the precoder set satisfying PUPC or PAPC forms an \textit{oblique manifold}. Further, we prove that the precoder sets form three different Riemannian submanifolds, from which we transform the constrained problems in Euclidean space to unconstrained ones on these Riemannian submanifolds. Sequentially, we derive Riemannian ingredients, including orthogonal projection, Riemannian gradient, Riemannian Hessian, retraction and vector transport, which are needed for precoder design in matrix manifold framework.
\subsection{WSR-maximization Precoder Design Problem in Euclidean Space}
We consider a single-cell massive MIMO system as shown in \figref{fig_system_model}. In the system, the BS equipped with $M_t$ antennas serves $U$ UTs and the $i$-th UT has $M_i$ antennas. The user set is denoted as $\cU=\left\{1,2,\cdots,U\right\}$ and the  set of the antennas at the BS side is denoted as  $\cA=\left\{1,2,\cdots,M_t\right\}$.  Let $\bx_i \in \bbC^{d_i}$ denote the signal transmitted to the $i$-th UT satisfying $\mathbb{E}\left\{\bx_i\bx_i^H\right\}=\bI_{d_i}$, where $d_i$ is the dimension of the signal transmitted to the $i$-th UT. $\bx=\sum_{i=1}^{U}\bP_i\bx_i$ is the transmitted signal for all the UTs. The received signal $\by_i$ of the $i$-th UT is given by
\begin{equation}\label{ReceivedSignal}
	\begin{aligned}
		\by_i=\overbrace{\bH_{i}\bP_{i}\bx_i}^{\mathrm{desired\ signal}}+\overbrace{\bH_{i}\sum_{ \ell \neq i,\ell=1}^{U}\bP_{\ell}\bx_{\ell}+\bz_i}^{\mathrm{interference\ plus\ noise }}, 
	\end{aligned}
\end{equation}
where $\bH_i$ is the complex channel matrix from the BS to the $i$-th UT,  $\bP_i$ is the corresponding precoding matrix, and $\bz_i$ denotes the independent and identically distributed (i.i.d.) complex circularly
symmetric Gaussian noise vector distributed as $\mathcal{CN}\left(0,\sigma_z^2\bI_{M_i}\right)$.  

\begin{figure}[t]
	\centering                 				
	\includegraphics[scale=0.65]{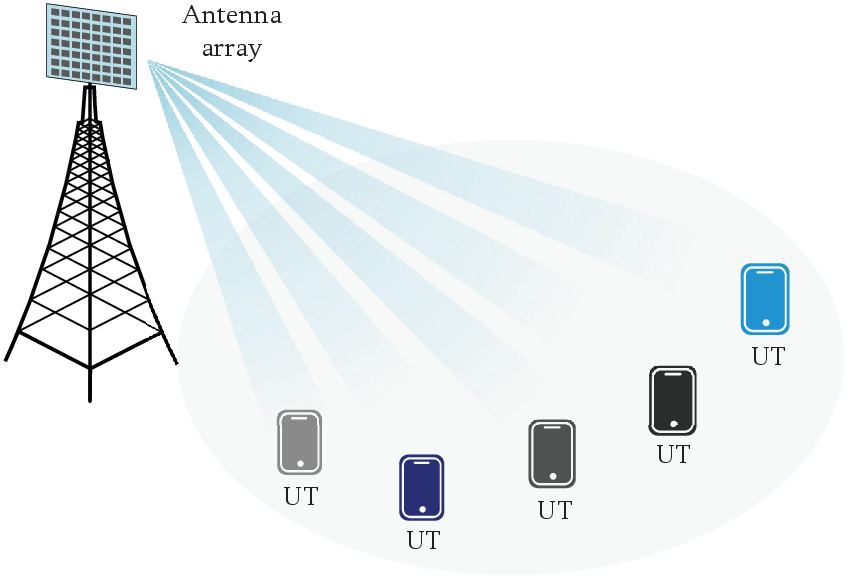}
	\caption{ The illustration of the massive MIMO system.}
	\label{fig_system_model}
\end{figure}

For simplicity, we assume that  the perfect channel state information (CSI) of the channel $\bH_i, \forall i\in\cU,$ is known at the BS side, and the perfect CSI of the effective channel $\bH_i\bP_i$ is available for the $i$-th UT via DL training. For the worst-case design, the aggregate interference plus noise $\bz_i^{\prime}=\bH_i\sum_{l\neq i}^{U}\bP_{\ell}\bx_{\ell}+\bz_i$ is treated as Gaussian noise. The covariance matrix of $\bz_i^{\prime}$ is as follows
\begin{equation}\label{R}
	\begin{aligned}
		\bR_i=\mathbb{E}\left\{\bz_i^{\prime}\left(\bz_i^{\prime}\right)^ H\right\}=\sigma_z^2\bI_{M_i}+\sum_{l\neq i}^{U}\bH_i\bP_{\ell}\bP_{\ell}^H\bH_i^H.
	\end{aligned}
\end{equation}
By assuming that $\bR_i$ is also known by UT $i$, the rate of user $i$ can be written as
\begin{equation}\label{RateOfUseri}
	\begin{aligned}
		\mathcal{R}_i=&\logdet{\bR_i+\bH_i\bP_i\bP_i^H\bH_i^H}-\logdet{\bR_i}\\
		=&\logdet{\bI_{d_i}+\bP_i^H\bH_i^H\left(\bR_i\right)^{-1}\bH_i\bP_i}.
	\end{aligned}
\end{equation}
The WSR-maximization precoder design problem can be formulated into a minimization problem  as
\begin{equation}\label{ProblemInVector}
    \begin{aligned}
	     \argmin{\bP_1,\cdots,\bP_U} f\left(\bP_1,\cdots,\bP_U\right) \ \ \mathrm{ s.t. } \  F\left(\bP_1,\bP_2,\cdots,\bP_U\right)=0,
   \end{aligned}
\end{equation}
where $f\left(\bP_1,\cdots,\bP_U\right)=-\sum_{i=1}^{U}w_i\mathcal{R}_i$ is the objective function with $w_i$ being the weighted factor of user $i$, and $F\left(\bP_1,\bP_2,\cdots,\bP_U\right)=0$ is the power constraint. We consider three specific power constraints in massive MIMO DL transmission including TPC, PUPC and PAPC. By stacking the precoding matrices of different users, we redefine the variable as 
\begin{equation}\label{StackMatrix}
	\begin{aligned}
		\bP&=\left(\bP_1,\bP_2,\cdots,\bP_U\right).
	\end{aligned}
\end{equation}
Let $P$, $P_i, \forall i\in\cU$, and $p_j, \forall j\in\cA$,  denote the total transmit power of the BS, the transmit power allocated for the $i$-th  UT and the power constraint of the $j$-th transmit antenna, respectively. The  WSR-maximization problem can be rewritten as
 \begin{equation}\label{Problem_TPC}
 		\begin{aligned}
 			\bP^{\star}=\left({\bP_1^{\star},\bP_2^{\star},\cdots,\bP_U^{\star}}\right)=\argmin{\bP} f\left(\bP\right) \ \ \mathrm{ s.t. } \ F\left(\bP\right)=0,
 		\end{aligned}
 \end{equation}
where $F\left(\bP\right)$ can be expressed as
\begin{subequations}
	\begin{align}
		\hat{F}\left(\bP\right)&=\mathrm{tr}\left\{ \bP^H\bP \right\}-P,\\
		\tilde{F}\left(\bP\right)&=\mathrm{tr}\left\{ \bP_i^H\bP_i \right\}-P_i,\forall\ i\in \cU,\\
		\bar{F}\left(\bP\right)&=\be_j^H\bP\bP^H\be_j -p_j, \forall\ j\in \cA
	\end{align}
\end{subequations}
for TPC, PUPC and PAPC, respectively. 
For simplicity, we assume that the power allocation process is done in the PUPC case and $\sum_{i=1}^{U}P_i= P$ without loss of generality. Besides, equal antenna power constraint is usually considered in practice to efficiently utilize the power amplifier capacity  of each antenna \cite{Equal} in the PAPC case, where the power constraint for each transmit antenna is $p_j=\frac{P}{M_t},\forall j\in\cA,$ and  $\bar{F}\left(\bP\right)$ can be reexpressed  as $\bar{F}\left(\bP\right)=\bI_{M_t} \odot\left(\bP\bP^H\right)=\frac{P}{M_t}\bI_{M_t}$.
\subsection{Problem Reformulation on Riemannian Submanifold}

In this subsection, we transform the constrained problems into unconstrained ones on three different Riemannian submanifolds.

From the perspective of manifold optimization, $\bP_i\in \bbC^{M_t\times d_i}$ belongs to the linear manifold $\bbC^{M_t \times d_i}$ naturally. Define the Riemannian metric as
\begin{equation}\label{RiemannianMetricOfProductManifold}
	\begin{aligned}
		g_{\bP_i}\left(\bxi_{\bP_i},\bzeta_{\bP_i}\right)=
		\Re\left\{\mathrm{tr}\left(\bxi_{\bP_i}^H\bzeta_{\bP_i} \right) \right\},
	\end{aligned}
\end{equation}
where $\bxi_{\bP_i}$ and $\bzeta_{\bP_i}$ are tangent vectors in tangent space $T_{\bP_i}{\bbC^{M_t \times d_i}}=\bbC^{M_t \times d_i}$. With the Riemannian metric \eqref{RiemannianMetricOfProductManifold}, $\bbC^{M_t \times d_i}$ forms a Riemannian manifold. In fact, $\bP=\left(\bP_1,\bP_2,\cdots,\bP_U\right)$ is a point on the product manifold \cite{Absil2009}
\begin{equation}
	\begin{aligned}
		\bbC^{M_t \times d_1}\times \bbC^{M_t \times d_2}\times \cdots \times \bbC^{M_t \times d_U}:=\cN.
	\end{aligned}
\end{equation}
Similarly, the tangent space of $\cN$ is given by \cite{Boumal2020}
\begin{equation}\label{product_tangent_space}
	\begin{aligned}
		T_{\bP}\cN=T_{\bP_1}\bbC^{M_t\times d_1}\times T_{\bP_2}\bbC^{M_t\times d_2}\times \cdots \times T_{\bP_U}\bbC^{M_t\times d_U},
	\end{aligned}
\end{equation}
of which the product Riemannian metric is defined as a direct sum \cite{Boumal2020} :
\begin{equation}\label{RiemannianProductMetric}
	\begin{aligned}
		g_{\bP}\left(\bxi_{\bP},\bzeta_{\bP}\right)=g_{\bP_1}\left(\bxi_{\bP_1},\bzeta_{\bP_1}\right)\oplus   \cdots \oplus g_{\bP_U}\left(\bxi_{\bP_U},\bzeta_{\bP_U}\right).
	\end{aligned}
\end{equation}
$\bxi_{\bP}=\left( \bxi_{\bP_1},\bxi_{\bP_2},\cdots,\bxi_{\bP_U} \right)$ and $\bzeta_{\bP}=\left( \bzeta_{\bP_1},\bzeta_{\bP_2},\cdots,\bzeta_{\bP_U} \right)$ are tangent vectors at point $\bP$ in $T_{\bP}\cN$. With the Riemannian metric defined in \eqref{RiemannianProductMetric}, $\cN$ is also a Riemannian manifold. 
Let
\begin{subequations}\label{sets}
	\begin{align}
		\widehat{\cM} &= \left\lbrace \bP \mid  \mathrm{tr}\left\{ \bP^H\bP \right\}=P \right\rbrace, \\
		\widetilde{\cM} &= \left\lbrace \bP \mid  \mathrm{tr}\left(\bP_i^H\bP_i\right) =P_i,\forall\ i\in \cU \right\rbrace, \\
		\overline{\cM} &= \left\lbrace \bP \mid  \bI_{M_t}\odot\left(\bP\bP^H\right) =\frac{P}{M_t}\bI_{M_t} \right\rbrace
	\end{align}		
\end{subequations}
denote the precoder sets satisfying TPC, PUPC and PAPC, respectively. We have the following theorem.

\begin{theorem}\label{theo_submanifold_TPC}
	
	The Frobenius norm of $\forall\hat{ \bP}\in\widehat{\cM}$ is a constant  and $\widehat{\cM}$ forms a sphere. The Frobenius norm of each component matrix $\tilde{\bP}_i$  of $\forall \tilde{\bP}\in\widetilde{\cM}$ is a constant and $\widetilde{\cM}$ forms a oblique manifold composed of $U$ spheres. The norm of each column of $\forall \bar{\bP}\in\overline{\cM}$ is a constant and  $\overline{\cM}$ forms an oblique manifold with $M_t$ spheres. $\widehat{\cM}$, $\widetilde{\cM}$ and $\overline{\cM}$ form three different Riemannian submanifolds of the product manifold $ \cN $.
	\begin{proof}
		The proof is provided in \appref{Proof_theo_Riemannian_Submanifold}.
	\end{proof}
\end{theorem}
 From \thref{theo_submanifold_TPC}, the constrained problems under TPC, PUPC and PAPC can be converted into unconstrained ones on manifolds as
\begin{subequations}
	\begin{align}
		\hat{\bP}^{\star}& = \argmin{\hat{\bP}\in \widehat{\cM}}f\left(\hat{\bP}\right),\\
		\tilde{\bP}^{\star}& = \argmin{\tilde{\bP}\in\widetilde{\cM}}f\left(\tilde{\bP}\right),\\
		\bar{\bP}^{\star}& = \argmin{\bar{\bP}\in\overline{\cM}}f\left(\bar{\bP}\right).
	\end{align}
\end{subequations}

In the case that TPC, PUPC and PAPC need to be satisfied simultaneously, the problem can be reformulated on the intersection of $\widetilde{\cM}$ and $\overline{\cM }$ and solved with the help of the von Neumann’s method of alternating projections proposed in \cite{Alter}. Precoder design with matrix manifold optimization is also investigated in \cite{RCG1}, where the matrix manifold optimization is used to  compute the related smallest eigenvalue and the corresponding eigenvector. So the problem is reformulated as the Rayleigh quotient minimization of generalized eigenvalue problems on the Riemannian quotient manifold, whose solution is not local optimal for the original problem.

\subsection{Riemannian Ingredients for Precoder Design}

In this subsection, we derive all the ingredients needed in manifold optimization for the three precoder design problems on Riemannian submanifolds.

First of all, we derive the Euclidean gradient of $ f\left( \bP \right)  $, which can be viewed as the Riemannian gradient of $ f\left( \bP \right)  $ on $\cN$. Recall that the Riemannian gradient of $ f\left( \bP \right)  $ on $\cN$ is the unique element $ \mathrm{grad} f\left( \bP\right) \in T_{\bP}\cN   $ that satisfies \eqref{RawGradient}.
Thus the Riemannian gradient $  \mathrm{grad} f \left( \bP\right)$ is identified from the directional derivative $ \mathrm{D} f \left( \bP \right) \left[ \bxi_{\bP}\right] $ with the Riemannian metric $ g_{\bP}\left(\cdot\right) $. Now let us define
\begin{equation} \label{A_i}
	\bA_i  = \bR_i^{-1}\bH_i\bP_i\in\bbC^{M_i\times d_i},
\end{equation}  
\begin{equation} \label{B_i}
	\bB_{i} = \bA_{i}\bC_{i}\bA_{i}^H\in\bbC^{M_i\times M_i},
\end{equation}
\begin{equation}\label{key}
	\bC_{i} =  \left( \bI_{d_{i}} + \bP_{i}^H\bH_{i}^H \bA_{i} \right)^{-1}\in\bbC^{d_i\times d_i}.
\end{equation}
\begin{theorem} \label{Theo_Euclidean_Gradient}
	\par The Euclidean gradient of $  f \left( \bP \right)  $ is
	\begin{equation}\label{Euclidean_Gradient_all}
		\mathrm{grad} f \left( \bP \right) = \left( \mathrm{grad} f \left( \bP_{1}\right) , \mathrm{grad} f \left( \bP_{2}\right), \cdots, \mathrm{grad} f \left( \bP_U\right) \right), 
	\end{equation} where 
	\begin{equation}\label{Euclidean_Gradient_single}
		\mathrm{grad} f \left( \bP_i\right) = -2\left( w_i \bH_i^H\bA_i\bC_i - \sum_{ \ell \neq i}w_{\ell}\bH_{\ell}^H\bB_{\ell}\bH_{\ell}\bP_i \right)
	\end{equation}
	is the Euclidean gradient on the $i $-th component submanifold $ \bbC^{M_t \times d_i} $. 
	\begin{proof}
		See the proof in \appref{App_Proof_Theo_Euclidean_Gradient}.
	\end{proof}
\end{theorem}

With \thref{Theo_Euclidean_Gradient}, we further derive the orthogonal projection, Riemannian gradient, retraction and vector transport for the three precoder design problems on Riemannian submanifolds.
\subsubsection{TPC}
\

From \eqref{RawComposition}, the tangent space $T_{\bP}\cN$ is decomposed into two orthogonal subspaces
\begin{equation}\label{Decomposition_TPC}
	\begin{aligned}
		T_{\bP}{\cN}=T_{\hat{\bP}}{\widehat{\cM}}\oplus N_{\hat{\bP}}{\widehat{\cM}}.
	\end{aligned}
\end{equation}
According to \eqref{ker}, the tangent space $T_{\hat{\bP}}{\widehat{\cM}}$ is given by
\begin{equation}\label{Tangent_Space_TPC}
	\begin{aligned}
		T_{\hat{\bP}}{\widehat{\cM}}=&\mathrm{ker}\left( \mathrm{D}\hat{F}\left(\bP\right) \right)\\
		=&\left\{  \bxi_{\bP}\in T_{\bP}\cN\mid \mathrm{tr}\left( \bP^H\bxi_{\bP}+\bxi_{\bP}^H\bP \right) =0\right\}.
	\end{aligned}
\end{equation}
The normal space of $\widehat{\cM}$ is the orthogonal complement of $T_{\hat{\bP}}{\widehat{\cM}}$ and can be expressed as
\begin{equation}\label{Normal_Space_TPC}
	\begin{aligned}
		N_{\hat{\bP}}\widehat{\cM} =& \left\{ \bxi_{\bP}\in T_{\bP}\cN \mid g_{\bP}\left(\bxi_{\bP},\bzeta_{\hat{\bP}}\right)=0,\ \forall\bzeta_{\hat{\bP}}\in T_{\hat{\bP}}\widehat{\cM} \right\}\\ =&\left\{\hat{\lambda}\bP\mid \hat{\lambda}\in\mathbb{R}\right\}.
	\end{aligned}
\end{equation}
Therefore, any $\bxi_{\bP} \in T_{\bP}\cN$ can be decomposed into two orthogonal parts as 
\begin{equation}\label{xiDecomposition_TPC}
	\begin{aligned}
		\bxi_{\bP}=\Pi_{T_{\widehat{\bP}}\widehat{\cM}}^{T_{\bP}\cN}\left( \bxi_{\bP}\right)+\Pi_{N_{\widehat{\bP}}\widehat{\cM}}^{T_{\bP}\cN}\left( \bxi_{\bP}\right),
	\end{aligned}
\end{equation}
where $\Pi_{T_{\widehat{\bP}}\widehat{\cM}}^{T_{\bP}\cN}\left( \bxi_{\bP}\right)$ and $\Pi_{N_{\widehat{\bP}}\widehat{\cM}}^{T_{\bP}\cN}\left( \bxi_{\bP}\right)$ represent the orthogonal projections of $\bxi_{\bP}$ onto $T_{\hat{\bP}}{\widehat{\cM}}$ and $N_{\hat{\bP}}{\widehat{\cM}}$, respectively.

\begin{lemma}\label{projection_submanifold_TPC}
	For any $ \bxi_{\bP} \in T_{\bP}\cN $,  the orthogonal projection $ \Pi_{T_{\widehat{\bP}}\widehat{\cM}}^{T_{\bP}\cN}\left( \bxi_{\bP}\right)$ is given by
	\begin{equation}\label{ProjectionTPC}
		\begin{aligned}
			\Pi_{T_{\widehat{\bP}}\widehat{\cM}}^{T_{\bP}\cN}\left( \bxi_{\bP}\right) = \bxi_{\bP} - \hat{\lambda} \bP, 		
		\end{aligned}					
	\end{equation}
	where
	\begin{equation}
		\begin{aligned}
			\hat{\lambda} =  \frac{1}{P}   \Re\left\lbrace  \mathrm{tr}\left(\bP^H\bxi_{\bP}\right) \right\rbrace. \label{ProjectionTPC_lambda}
		\end{aligned}
	\end{equation}

\end{lemma}
\begin{proof}
	See \appref{App_proof_prop_projection_TPC} for the proof.
\end{proof}

\par Given the orthogonal projection $ \Pi_{T_{\hat{\bP}}\widehat{\cM}}^{T_{\bP}\cN}\left( \bxi_{\bP}\right) $, we derive the Riemannian gradient and Riemannian Hessian of $ f\big( \hat{\bP} \big)  $ subsequently. 

\begin{theorem}\label{Theo_Riemannian_Gradient_sub_TPC}
	The Riemannian gradient of $ f\big( \hat{\bP}\big)$ is 
	\begin{equation}\label{Riemannian_gradient_TPC}
		\begin{aligned}
			\mathrm{grad} f\left( \hat{\bP}\right) = \mathrm{grad} f \left( \bP\right) -\hat{\lambda}_1 \bP,
		\end{aligned}		  
	\end{equation}
	 where
	\begin{equation}\label{hatlambda1}
		\hat{\lambda}_1 =  \frac{1}{P}   \Re\left\lbrace  \mathrm{tr}\left(\bP^H\mathrm{grad} f \left( \bP\right)\right) \right\rbrace	 .
	\end{equation}
	The Riemannian Hessian of $  f \big( \hat{\bP} \big)  $ on $\widehat{\cM}$ is given by
	\begin{equation}\label{Riemannian_Hessian_all}
		\begin{aligned}
			\mathrm{Hess}f(\hat{\bP})[\bxi_{\hat{\bP}}]=&\left(\mathrm{Hess}f(\hat{\bP}_1)[\bxi_{\hat{\bP}_1}],\cdots,\mathrm{Hess}f(\hat{\bP}_U)[\bxi_{\hat{\bP}_U}]\right),
		\end{aligned}
	\end{equation}
	where
	\begin{equation}\label{hessian_TPC}
		\begin{aligned}
			&\mathrm{Hess}f(\hat{\bP}_i)[\bxi_{\hat{\bP}_i}]=\mathrm{Dgrad}f(\hat{\bP}_i)[\bxi_{\hat{\bP}_i}]-\hat{\lambda}_2\bP_i \\
			&=\mathrm{Dgrad}f(\bP_i)[\bxi_{\bP_i}] - \mathrm{D}\hat{\lambda}_1[\bxi_{\bP_i}]\bP_i-\hat{\lambda}_1 \bxi_{\bP_i}-\hat{\lambda}_2\bP_i ,
		\end{aligned}
	\end{equation}
	\begin{equation}\label{hessian_mu}
		\begin{aligned}
			\hat{\lambda}_2=\frac{1}{P}\Re\left\{\bP^H\mathrm{Dgrad}f(\hat{\bP})[\bxi_{\hat{\bP}}]\right\}
		\end{aligned}.
	\end{equation}
	\begin{proof}
		The proof and the details of Riemannian gradient and Riemannian Hessian in the TPC case are provided in \appref{App_Proof_Theo_TPC}.
	\end{proof}
\end{theorem}

From \thref{theo_submanifold_TPC}, $\widehat{\cM}$ is a sphere, whose retraction can be defined as
\begin{equation}\label{retraction_TPC}
	\begin{aligned}
		&R_{\hat{\bP}}\left( \bxi_{\hat{\bP} }\right): T_{\hat{\bP}}\widehat{\cM} \to \widehat{\cM} : \\
		&\bxi_{\hat{\bP} }  \mapsto  \frac{\sqrt{P}\left(\hat{\bP}+\bxi_{\hat{\bP}}\right)}{\sqrt{\mathrm{tr}\left(\left(\hat{\bP}+\bxi_{\hat{\bP}}\right)^H\left(\hat{\bP}+\bxi_{\hat{\bP}}\right)\right)}}=\hat{\gamma} \left(\hat{\bP}+\bxi_{\hat{\bP}}\right).
 	\end{aligned}	 
\end{equation}
In Riemannian submanifold, the vector transport $\cT_{\boldsymbol{\eta}_{\widehat{\bP}}}\left(\bxi_{\hat{\bP}}\right)$ can be achieved like \eqref{RawVectorTransport} by orthogonally projecting $\bxi_{\hat{\bP}}\in T_{\hat{\bP}}\widehat{\cM}$ onto $T_{R_{\hat{\bP}}\left(\boldsymbol{\eta}_{\hat{\bP}}\right)}\widehat{\cM}$. Let $\hat{ \bP}^{\mathrm{Re}}= R_{\hat{\bP}}\left(\boldsymbol{\eta}_{\hat{\bP}}\right)$, we have
\begin{equation}\label{VT_TPC}
	\begin{aligned}
		\cT_{\boldsymbol{\eta}_{\widehat{\bP}}}\left(\bxi_{\hat{\bP}}\right)=\Pi_{T_{\widehat{ \bP}^{\mathrm{Re}}}\widehat{\cM}}^{T_{\widehat{\bP}}\widehat{\cM}}\left( \bxi_{\hat{\bP}}\right),
	\end{aligned}
\end{equation}
where the projection can be derived in a similar way with \lmref{projection_submanifold_TPC}.

\subsubsection{PUPC}
\

 From \eqref{ker}, the tangent space $T_{\tilde{\bP}}{\widetilde{\cM}}$ is 
\begin{equation}
	\begin{aligned}
		T_{\tilde{\bP}}{\widetilde{\cM}}
		=&\left\{  \bxi_{\bP}\in T_{\bP}\cN\mid \mathrm{tr}\left(\bP_i^H\bxi_{\bP_i}+\bxi_{\bP_i}^H\bP_i\right)=0,\forall i\in \cU \right\}.
	\end{aligned}
\end{equation}
The normal space of $\widetilde{\cM}$ can be expressed as
\begin{equation}
	\begin{aligned}
		N_{\tilde{\bP}}\widetilde{\cM} 
		=&\left\{\bP\tilde{\bLambda}\mid \tilde{\bLambda}\in\widetilde{\mathcal{D}}\right\}		,
	\end{aligned}
\end{equation}
where $\widetilde{\mathcal{D}}=\left\{\mathrm{blkdiag}\left(\bD_1,\cdots,\bD_U\right)\mid \bD_i=a_i\bI_{d_i}, a_i\in\bbR  \right\}$ is the block diagonal matrix subset whose dimension is $U$. With the normal space, we can derive the orthogonal projection of a tangent vector from $T_{\bP}\cN$ to  $T_{\tilde{\bP}}\widetilde{\cM}$.
\begin{lemma}\label{projection_submanifold_PUPC}
	For any $ \bxi_{\bP} \in T_{\bP}\cN $,  the orthogonal projection $ \Pi_{T_{\widetilde{\bP}}\widetilde{\cM}}^{T_{\bP}\cN}\left( \bxi_{\bP}\right)$ is given by
	\begin{subequations}\label{key}
		\begin{align}
			\Pi_{T_{\widetilde{\bP}}\widetilde{\cM}}^{T_{\bP}\cN}\left( \bxi_{\bP}\right) = \bxi_{\bP} -  \bP\tilde{\bLambda},\\
			\big[\tilde{\bLambda}\big]_i =  \frac{1}{P_i} \Re\left\lbrace  \bP_i^H\bxi_{\bP_i} \right\rbrace\bI_{d_i}.
		\end{align}				
	\end{subequations}
	
\end{lemma}
\begin{proof}
	The proof is similar with \appref{App_proof_prop_projection_TPC} and thus omitted for brevity.
\end{proof}

The Riemannian gradient and Hessian in $T_{\tilde{\bP}}\widetilde{\cM}$ can be obtained by projecting the Riemannian gradient and Hessian in $T_{\bP}\cN$ onto $T_{\tilde{\bP}}\widetilde{\cM}$.
\begin{theorem} \label{Theo_Riemannian_Gradient_sub_PUPC}
	\par The Riemannian gradient of $ f\big( \tilde{\bP}\big)$ is 
	\begin{equation}\label{Riemannian_gradient_PUPC}
		\begin{aligned}
			 \mathrm{grad} f\big( \tilde{\bP}\big) = \mathrm{grad} f \left( \bP\right) -\bP\tilde{\bLambda}_1 , 
		\end{aligned}
	\end{equation}
	 where
	\begin{equation}\label{tildeLambda1}
		\big[\tilde{\bLambda}_1\big]_i =   \frac{1}{P_i} \Re\Big\{\mathrm{tr}\left(  \bP_i^H\mathrm{grad} f\left( \bP_i\right)\right)\Big\} \bI_{d_i} 	 .
	\end{equation}

	The Riemannian Hessian of $  f \big( \tilde{\bP} \big)  $ is
	\begin{equation}\label{hessian_PUPC}
		\begin{aligned}
		    &\mathrm{Hess}f(\tilde{\bP})[\bxi_{\bP}]=\left(\right.\mathrm{Hess}f(\tilde{\bP}_1)[\bxi_{\tilde{\bP}_1}],\cdots,\mathrm{Hess}f(\tilde{\bP}_U)[\bxi_{\tilde{\bP}_U}]\left.\right)\\
		    &=\mathrm{Dgrad}f(\bP)[\bxi_{\bP}] - \bP\mathrm{D}\tilde{\bLambda}_1[\bxi_{\bP}]- \bxi_{\bP}\tilde{\bLambda}_1-\bP\tilde{\bLambda}_2,
		\end{aligned}
	\end{equation}
	where
	\begin{equation}\label{hessian_mu}
		\begin{aligned}
			\big[\tilde{\bLambda}_2\big]_i=\frac{1}{P_i}\Re\left\{\mathrm{tr}\left(\bP^H\left(\mathrm{Dgrad}f(\tilde{\bP})[\bxi_{\tilde{\bP}}]\right)\right)\right\}\bI_{d_i}.
		\end{aligned}
	\end{equation}
	\begin{proof}
		The proof of the Riemannian gradient is similar with \appref{App_Proof_Theo_TPC} and thus omitted for brevity. The proof and  details of the Riemannian Hessian are provided in \appref{App_Proof_Theo_Riemannian_Hessian_sub_PUPC}.
	\end{proof}
\end{theorem}

From \thref{theo_submanifold_TPC}, $\widetilde{\cM}$ is a product of $U$ spheres called oblique manifold, whose retraction can be defined by scaling \cite{ObliqueRetraction}.  Define
\begin{equation}\label{tildegamma}
	\begin{aligned}
		\tilde{\gamma}_i &=  \frac{\sqrt{P_i}}{\sqrt{\mathrm{tr}\left(\left(\tilde{\bP}_i+\bxi_{\tilde{\bP}_i}\right)^H\left(\tilde{\bP}_i+\bxi_{\tilde{\bP}_i}\right)\right)}},
	\end{aligned}
\end{equation}
and let $\tilde{\bGamma}_i=\tilde{\gamma}_i\bI_{d_{i}}$ and $\tilde{\bGamma}=\mathrm{blkdiag}\left(  \tilde{\bGamma}_1,\cdots,\tilde{\bGamma}_U \right)\in \widetilde{\mathcal{D}}$. For any $ \tilde{\bP} \in \widetilde{\cM} $ and $ \bxi_{\tilde{\bP} } \in T_{\tilde{\bP}}\widetilde{\cM} $, the mapping
\begin{equation}\label{retraction_PUPC}
	R_{\tilde{\bP}}\left( \bxi_{\tilde{\bP} }\right): T_{\tilde{\bP}}\widetilde{\cM} \to \widetilde{\cM} :  \bxi_{\tilde{\bP} }  \mapsto  \left(\tilde{\bP}+\bxi_{\tilde{\bP}}\right)\tilde{\bGamma}
\end{equation}
is a retraction for $ \widetilde{\cM} $ at $ \tilde{\bP} $.

Similarly, the vector transport $\cT_{\boldsymbol{\eta}_{\tilde{\bP}}}\left(\bxi_{\tilde{\bP}}\right)$ can be achieved by  projecting $\bxi_{\tilde{\bP}}\in T_{\tilde{\bP}}\widetilde{\cM}$ onto $T_{R_{\tilde{\bP}}\left(\boldsymbol{\eta}_{\tilde{\bP}}\right)}\widetilde{\cM}$ orthogonally. Let $\tilde{\bP}^{\mathrm{Re}}$ represent $R_{\tilde{\bP}}\left(\boldsymbol{\eta}_{\tilde{\bP}}\right)$, we have
\begin{equation}\label{VT_PUPC}
	\begin{aligned}
		\cT_{\boldsymbol{\eta}_{\widetilde{\bP}}}\left(\bxi_{\tilde{\bP}}\right)=\Pi_{T_{\widetilde{\bP}^{\mathrm{Re}}}\widetilde{\cM}}^{T_{\widetilde{\bP}}\widetilde{\cM}}\left( \bxi_{\tilde{\bP}}\right),
	\end{aligned}
\end{equation}
where the orthogonal projection operator can be obtained in a similar way with \lmref{projection_submanifold_PUPC}.

\subsubsection{PAPC}
\

From \eqref{ker}, the tangent space $T_{\bar{\bP}}{\overline{\cM}}$ is 
\begin{equation}
	\begin{aligned}
		T_{\bar{\bP}}{\overline{\cM}}
		&=\left\{  \bxi_{\bP}\in T_{\bP}\cN\mid \bI_{M_t}\odot \left(\bP\bxi_{\bP}^H+\bxi_{\bP}\bP^H\right)=0 \right\}.
	\end{aligned}
\end{equation}
The normal space of $\overline{\cM}$ is given by
\begin{equation}
	\begin{aligned}
		N_{\bar{\bP}}\overline{\cM} 
		 =&\left\{\bar{\bLambda}\bP\mid \bar{\bLambda}\in\overline{\mathcal{D}}\right\},
	\end{aligned}
\end{equation}
where $\overline{\mathcal{D}}=\left\{\mathrm{diag}\left( \bw \right) \mid \bw\in\mathbb{R}^{M_t}\right\}$ is the diagonal matrix set.
Similarly, we can derive $\Pi_{T_{\overline{\bP}}\overline{\cM}}^{T_{\bP}\cN}\left( \bxi_{\bP}\right)$ with the closed-form normal space as follows.

\begin{lemma}\label{projection_submanifold_PAPC}
	For any $ \bxi_{\bP} \in T_{\bP}\cN $,  the orthogonal projection $ \Pi_{T_{\bar{\bP}}\overline{\cM}}^{T_{\bP}\cN}\left( \bxi_{\bP}\right) $ is given by
	\begin{subequations}\label{key}
		\begin{align}
			\Pi_{T_{\overline{\bP}}\overline{\cM}}^{T_{\bP}\cN}\left( \bxi_{\bP}\right)= \bxi_{\bP} - \bar{\bLambda} \bP,\\
			\bar{\bLambda} =  \frac{M_t}{P} \bI_{M_t} \odot  \Re\left\lbrace  \bP\bxi_{\bP}^H \right\rbrace.
		\end{align}		
	\end{subequations}
\end{lemma}
\begin{proof}
	The proof is similar with \appref{App_proof_prop_projection_TPC} and thus omitted for brevity.
\end{proof}
With $\Pi_{T_{\overline{\bP}}\overline{\cM}}^{T_{\bP}\cN}\left( \cdot\right)$, the Riemannian gradient and Hessian in  $T_{\bar{\bP}}\overline{\cM}$ are provided in the following theorem.
\begin{theorem} \label{Theo_Riemannian_Gradient_sub_PAPC}
	\par The Riemannian gradient of $f\left( \bar{\bP}\right)$ is 
	\begin{equation}\label{Riemannian_gradient_PAPC}
		\begin{aligned}
			\mathrm{grad} f\left( \bar{\bP}\right) = \mathrm{grad} f \left( \bP\right) -\bar{\bLambda}_1 \bP,
		\end{aligned}
	\end{equation}
	 where
	\begin{equation}\label{barLambda1}
		\bar{\bLambda}_1 =  \frac{M_t}{P} \bI_{M_t} \odot  \Re \left\lbrace   \bP\left( \mathrm{grad} f\left( \bP\right)\right) ^H \right\rbrace. 	 
	\end{equation}

	The Riemannian Hessian of $  f \left( \bar{\bP} \right)  $ is
	\begin{equation}\label{Riemannian_Hessian_all}
		\begin{aligned}
			\mathrm{Hess}f(\bar{\bP})[\bxi_{\bar{\bP}}]=&\left(\mathrm{Hess}f(\bar{\bP}_1)[\bxi_{\bar{\bP}_1}],\cdots,\mathrm{Hess}f(\bar{\bP}_U)[\bxi_{\bar{\bP}_U}]\right),
		\end{aligned}
	\end{equation}
	where
	\begin{equation}\label{hessian_PAPC}
		\begin{aligned}
			&\mathrm{Hess}f(\bar{\bP}_i)[\bxi_{\bar{\bP}_i}]=\\
			&\mathrm{Dgrad}f(\bP_i)[\bxi_{\bP_i}]- \mathrm{D}\bar{\bLambda}_1[\bxi_{\bP_i}]\bP_i-\bar{\bLambda}_1 \bxi_{\bP_i}-\bar{\bLambda}_2\bP_i,
		\end{aligned}
	\end{equation}
	\begin{equation}\label{hessian_mu}
		\begin{aligned}
			\bar{\bLambda}_2=\frac{M_t}{P}\bI_{M_t}\odot\sum_{\ell=1}^{K}\Re\left\{\bP_{\ell}\left(\mathrm{Dgrad}f(\bar{\bP}_{\ell})[\bxi_{\bar{\bP}_{\ell}}]\right)^H\right\}.
		\end{aligned}
	\end{equation}
	\begin{proof}
		The proof of the Riemannian gradient is similar with \appref{App_Proof_Theo_TPC} and thus omitted for brevity. The proof and details of the Riemannian Hessian are provided in \appref{App_Proof_Theo_Riemannian_Hessian_sub_PAPC}.
	\end{proof}
\end{theorem}

Like PUPC, the retraction here can also be defined by scaling \cite{ObliqueRetraction}.  Let 
\begin{equation}\label{bargamma}
	\begin{aligned}
		\bar{\bGamma} &=  \left(  \frac{M_t}{P} \bI_{M_t} \odot \left(\bar{\bP} + \bxi_{\bar{\bP} } \right)  \left(\bar{\bP} + \bxi_{\bar{\bP} } \right)^H \right)^{-\frac{1}{2}}. 
	\end{aligned}
\end{equation}
For any $ \bar{\bP} \in \overline{\cM} $ and $ \bxi_{\bar{\bP} } \in T_{\bar{\bP}}\overline{\cM} $, the mapping
\begin{equation}\label{retraction_PAPC}
	R_{\bar{\bP}}\left( \bxi_{\bar{\bP} }\right): T_{\bar{\bP}}\overline{\cM} \to \overline{\cM} :  \bxi_{\bar{\bP} }  \mapsto  \bar{\bGamma}\left(\bar{\bP}+\bxi_{\bar{\bP}}\right)
\end{equation}
is a retraction for $ \overline{\cM} $ at $ \bar{\bP} $.

The vector transport $\cT_{\boldsymbol{\eta}_{\bar{\bP}}}\left(\bxi_{\bar{\bP}}\right)$ likewise can be achieved by orthogonally projecting $\bxi_{\bar{\bP}}\in T_{\bar{\bP}}\overline{\cM}$ onto $T_{R_{\bar{\bP}}\left(\boldsymbol{\eta}_{\bar{\bP}}\right)}\overline{\cM}$. Let $\bar{\bP}^{\mathrm{Re}}$ denote $R_{\bar{\bP}}\left(\boldsymbol{\eta}_{\bar{\bP}}\right)$, we have
\begin{equation}
	\begin{aligned}
		\cT_{\boldsymbol{\eta}_{\overline{\bP}}}\left(\bxi_{\bar{\bP}}\right)=\Pi_{T_{\overline{\bP}^{\mathrm{Re}}}\overline{\cM}}^{T_{\overline{\bP}}\overline{\cM}}\left( \bxi_{\bar{\bP}}\right),
	\end{aligned}
\end{equation}
which can be derived in a similar way with \lmref{projection_submanifold_PAPC}. 
\section{Riemannian Methods for Precoder Design}\label{RiemannianAlgorithms}
With the Riemannian ingredients derived in \secref{Precoding}, we propose three precoder design methods using the RSD, RCG and RTR in this section. There is no inverse of the large dimensional matrix in the proposed Riemannian methods during the iterations, thereby enabling significant savings in computational resources. For the same power constraint, the computational complexities of the RSD or RCG method are nearly the same and  are lower than those of the RTR and comparable methods.
\subsection{Riemannian Gradient Methods}  

The gradient descent method is one of the most well-known and efficient line search methods in Euclidean space for unconstrained problems. For the Riemannian manifold, we update the point through retraction to ensure the updated point is still on the manifold and preserve the search direction. For notational clarity, we use the superscript $k$ to stand for the outer iteration. From \eqref{retraction_TPC}, \eqref{retraction_PUPC} and \eqref{retraction_PAPC}, the update formula on $\widehat{\cM}$, $\widetilde{\cM }$ and $\overline{\cM}$ are given, respectively, by
\begin{subequations}\label{Retraction}
	\begin{align}
		\hat{\bP}_{i}^{k+1}&=\hat{\gamma}^k\left(\hat{\bP}_i^k+\alpha^k\hat{\boldsymbol{\eta}}_i^k\right),\\
		\tilde{\bP}_{i}^{k+1}&=\tilde{\gamma}_i^k\left(\tilde{\bP}_{i}^k+\alpha^k\tilde{\boldsymbol{\eta}}_i^k\right),\\
		\bar{\bP}_{i}^{k+1}&=\bar{\bGamma}^k\left(\bar{\bP}_i^k+\alpha^k\bar{\boldsymbol{\eta}}_i^k\right),
	\end{align}
\end{subequations}
where $\alpha^k$ is the step length and $\hat{\boldsymbol{\eta}}_i^{k}$, $\tilde{\boldsymbol{\eta}}_i^{k}$ and $\bar{\boldsymbol{\eta}}_i^{k}$ are the search directions of user $i$. With \eqref{Retraction} and the Riemannian gradient, the RSD method is available but converges slowly. The RCG method  provides a remedy to this drawback by modifying the search direction, which calls for the sum of the current Riemannian gradient and the previous search direction \cite{RCG1,RCG2}. To add the elements in different tangent spaces, vector transport defined in \eqref{RawVectorTransport} is utilized. We use $\cM$ to represent $\widehat{\cM}$, $\widetilde{\cM}$ or $\overline{\cM}$. The search direction of the RCG method in the $k$-th iteration is
\begin{equation}\label{UpdateSearchDirection}
	\begin{aligned}
 \boldsymbol{\eta}^k=-\mathrm{grad}f\left(\bP^k\right)+\beta^k\cT^{\cM}_{\alpha^{k-1}{{\boldsymbol{\eta}}}^{k-1}}\left({{\boldsymbol{\eta}}}^{k-1}\right),
	\end{aligned}
\end{equation}
where ${{\beta}}^k$ is chosen as the Fletcher-Reeves parameter:
\begin{equation}\label{beta}
	\begin{aligned}
		{{\beta}}^k=\frac{g^{\cM}_{{{\bP}}^k}\left(\mathrm{grad}f\left({{\bP}}^k\right),\mathrm{grad}f\left({{\bP}}^k\right)\right)}{g^{\cM}_{{{\bP}}^{k-1}}\left(\mathrm{grad}f\left({{\bP}}^{k-1}\right),\mathrm{grad}f\left({{\bP}}^{k-1}\right)\right)}.
	\end{aligned}
\end{equation}

RSD and RCG both need to compute the Riemannian gradient, which is made up of the Euclidean gradient \eqref{Euclidean_Gradient_all} and the orthogonal projection. 
Let ${{\bV}}_{i,\ell}^k=\bH_{i}{{\bP}}_{\ell}^k$ and ${{\bU}}_{i,\ell}^k=\bH_{i}{{\boldsymbol{\eta}}}_{\ell}^k, \forall i, \ell\in\cU$, $\bR_i$ in the $k$-th iteration can be written as 
\begin{equation}
	\begin{aligned}
		\bR_i^k=&\sigma_z^2\bI_{M_i}+\sum_{\ell\neq i}^{U}{{\bV}}_{i,\ell}^k\left({{\bV}}_{i,\ell}^k\right)^H, \forall  i\in\cU.\\
	\end{aligned}
\end{equation}
Then the Euclidean gradient of the $i$-th UT, $\forall i\in\cU$, in the $k$-th iteration can be written as
\begin{equation}\label{al_gradient_TPC}
	\begin{aligned}
		&\mathrm{grad}f\left({{\bP}}_{i}^k\right)=-2w_i \bH_i^H\bR_i^{-1}{{\bV}}_{i,i}^k\left(\bI_{d_i}+\left({{\bV}}_{i,i}^k\right)^H\bR_i^{-1}{{\bV}}_{i,i}^k\right)^{-1}\\
		&+2\sum_{ \ell \neq i}w_{\ell}\bH_{\ell}^H\bR_{\ell}^{-1}{{\bV}}_{\ell,\ell}^k\left(\bI_{d_{\ell}}+\left({{\bV}}_{\ell,\ell}^k\right)^H\bR_{\ell}^{-1}{{\bV}}_{\ell,\ell}^k\right)^{-1}\times\\
		&\left({{\bV}}_{\ell,\ell}^k\right)^H\bR_{\ell}^{-1}{{\bV}}_{\ell,i}^k,
	\end{aligned}
\end{equation}
where ${{\bV}}_{i,\ell}^k,\forall i, \ell\in\cU$, can be obtained from \eqref{Retraction} and expressed as 
\begin{subequations}\label{V}
	\begin{align}
		{{\bV}}_{i,\ell}^k&=\hat{\gamma}^{k-1}\left(\hat{\bV}_{i,\ell}^{k-1}+\alpha^{k-1}\hat{\bU}_{i,\ell}^{k-1}\right)\ \ \  \text{for}\  \widehat{\cM}\label{V_1},\\
		{{\bV}}_{i,\ell}^k&=\tilde{\gamma}_{\ell}^{k-1}\left(\tilde{\bV}_{i,\ell}^{k-1}+\alpha^{k-1}\tilde{\bU}_{i,\ell}^{k-1}\right)\ \ \ \text{for}\ \widetilde{\cM}\label{V_2},\\
		{{\bV}}_{i,\ell}^k&=\bH_{i}\bar{\bGamma}^{k-1}\left(\bar{\bP}_{\ell}^{k-1}+\alpha^{k-1}\bar{\boldsymbol{\eta}}_{\ell}^{k-1}\right)\ \text{for}\ \overline{\cM}\label{V_3}.
	\end{align}
\end{subequations}  
We can see that ${{\bV}}_{i,\ell}^k,\forall i, \ell\in\cU$, can be obtained directly from ${{\bV}}_{i,\ell}^{k-1}$ and ${{\bU}}_{i,\ell}^{k-1}$ for $\widehat{\cM}$ and $\widetilde{\cM}$. On the contrary, ${{\bV}}_{i,\ell}^k$ needs to be computed in each iteration for $\overline{\cM}$.

The step length $\alpha^k$ in \eqref{V} can be obtained through the backtracking method \cite{Numerical}, where the objective function needs to be evaluated to ensure sufficient decrease. Compared with the projected steepest descent method \cite{projected}, only one step length need to be searched in our proposed RSD and RCG methods. For notational clarity, we use the superscript pair $(k,n)$ to denote the $n$-th inner iteration for searching for the step length during the $k$-th outer iteration. As the precoder $\bP^k$ and the search direction $\boldsymbol{\eta}^k$ are fixed when searching for the step length, the objective function in the $(k,n)$-th iteration can be viewed as a function of $\alpha^{k,n-1}$ and written as
\begin{equation}\label{al_f}
	\begin{aligned}
		&\phi\left(\alpha^{k,n-1}\right)=f\left(R_{{{\bP}}^{k,n-1}}\left(\alpha^{k,n-1}{{\boldsymbol{\eta}}}^{k,n-1}\right)\right)=\\
		&-\sum_{i=1}^U\logdet{\sigma_z^2\bI_{M_i}+\sum_{\ell}^U{{\bV}}_{i,\ell}^{k,n}\left({{\bV}}_{i,\ell}^{k,n}\right)^H}\\
		&+\sum_{i=1}^U\logdet{\sigma_z^2\bI_{M_i}+\sum_{\ell\neq i}^{U}{{\bV}}_{i,\ell}^{k,n}\left({{\bV}}_{i,\ell}^{k,n}\right)^H},
	\end{aligned}
\end{equation}
where ${{\bV}}_{i,\ell}^{k,n}={{\bV}}_{i,\ell}^{k,n}\left(\alpha^{k,n-1}\right), \forall i,\ell\in\cU$, is viewed as a function of $\alpha^{k,n-1}$. Note that ${{\bV}}_{i,\ell}^{k,n}\left(\alpha^{k,n-1}\right)$ can be derived in the same way as ${{\bV}}_{i,\ell}^k$ in \eqref{V} and are given by  
\begin{subequations}\label{V_inner}
	\begin{align}
		{{\bV}}_{i,\ell}^{k,n}\left(\alpha^{k,n-1}\right)&=\hat{\gamma}\left(\alpha^{k,n-1}\right)\left(\hat{\bV}_{i,\ell}^{k}+\alpha^{k,n-1}\hat{\bU}_{i,\ell}^{k}\right)\ \  \text{for}\  \widehat{\cM}\label{V_inner_1},\\
		{{\bV}}_{i,\ell}^{k,n}\left(\alpha^{k,n-1}\right)&=\tilde{\gamma}_{\ell}\left(\alpha^{k,n-1}\right)\left(\tilde{\bV}_{i,\ell}^{k}+\alpha^{k,n-1}\tilde{\bU}_{i,\ell}^{k}\right)\  \text{for}\ \widetilde{\cM}\label{V_inner_2},\\
		{{\bV}}_{i,\ell}^{k,n}\left(\alpha^{k,n-1}\right)&=\bH_{i}\bar{\bGamma}\left(\alpha^{k,n-1}\right)\left(\bP_{\ell}^{k}+\alpha^{k,n-1}\bar{\boldsymbol{\eta}}_{\ell}^{k}\right)\ \  \text{for}\ \overline{\cM}\label{V_inner_3}.
	\end{align}
\end{subequations}  
 Similarly,  ${{\bV}}_{i,\ell}^{k,n}\left(\alpha^{k,n-1}\right)$ can be obtained directly from ${{\bV}}_{i,\ell}^{k}$ and ${{\bU}}_{i,\ell}^{k}$ for $\widehat{\cM}$ and $\widetilde{\cM}$, while it needs to be computed once per inner iteration for $\overline{\cM}$. Note that $\hat{\gamma}\left(\alpha^{k,n-1}\right)$, $\tilde{\gamma}_{\ell}\left(\alpha^{k,n-1}\right)$ and $\bar{\bGamma}\left(\alpha^{k,n-1}\right)$ defined in \eqref{retraction_TPC}, \eqref{tildegamma} and \eqref{bargamma} are viewed as functions of $\alpha^{k,n-1}$ here.

The procedure of RSD and RCG methods for precoder design is provided in \alref{RCG}. RSD method is available by updating the search direction ${{\boldsymbol{\eta}}}^{k+1}$ with $-\mathrm{grad}f\left({{\bP}}^{k+1}\right)$, while RCG method is achieved by updating the search direction ${{\boldsymbol{\eta}}}^{k+1}$ with \eqref{UpdateSearchDirection}. Typically,  $r=0.5$ and $c=10^{-4}$. The analyses of the convergence properties of the RSD and RCG methods can be found in \appref{App_Proof_Convergence}. 

\begin{algorithm}[h]  
	\caption{RSD and RCG methods for precoder design} 
	\label{RCG}  
	\renewcommand{\algorithmicensure}{\textbf{Output:}}
	\begin{algorithmic}[1] 
		\REQUIRE Riemannian submanifold $\cM$; Riemannian metric $g_{{{\bP}}}\left(\cdot\right)$; Real-valued function $f$;  \\
		Retraction $R_{{{\bP}}}\left(\cdot\right)$; Vector transport $\cT_{{{\boldsymbol{\eta}}}}\left(\cdot\right)$; initial step length $\alpha^0>0$; $r\in \left(0,1\right)$; $c\in \left(0,1\right)$
		\renewcommand{\algorithmicrequire}{\textbf{Input:}}
		\REQUIRE Initial point ${{\bP}}^0$;
        \REPEAT
		\STATE Get ${{\bV}}_{i,\ell}^k,\forall i,\ell\in\cU$ with \eqref{V}.
		\STATE  Compute Euclidean gradient $\mathrm{grad}f({{\bP}}_{i}^k), \forall i,\ell\in\cU$ with \eqref{al_gradient_TPC}. 
		\STATE Get Riemannian gradient $\mathrm{grad}f({{\bP}}^{k})$ with \eqref{Riemannian_gradient_TPC}, \eqref{Riemannian_gradient_PUPC} or \eqref{Riemannian_gradient_PAPC}.
		\STATE Update the search direction ${{\boldsymbol{\eta}}}^{k}$ and compute ${{\bU}}_{i,\ell}^k,\forall i,\ell\in\cU$.
		\WHILE{$\phi\left(\alpha^{k,n-1}\right)-f\left({{\bP}}^{k}\right)$ $\geq c\times g^{\cM}_{{{\bP}}^k}\left(\mathrm{grad}f\left({{\bP}}^k\right),\right.$ $\left. \alpha^{k,n-1} {{\boldsymbol{\eta}}}^k\right)$}
		\STATE  Set $\alpha^{k,n}\leftarrow r\alpha^{k,n-1}$ with	  $\alpha^{k,0} = \alpha^0$.
		\STATE Get ${{\bP}}^{k,n+1}$ with \eqref{Retraction} and ${{\bV}}_{i,\ell}^{k,n}\left(\alpha^{k,n}\right),\forall i,\ell\in\cU$ with \eqref{V_inner}.
		\STATE Get $\phi\left(\alpha^{k,n}\right)$ with \eqref{al_f},   $n\leftarrow n+1$.
		\ENDWHILE
		\STATE  Set ${{\bP}}^{k+1}\leftarrow{{\bP}}^{k,n}$, $k\leftarrow k+1$.
		\UNTIL{convergence}
	\end{algorithmic}  
\end{algorithm}

\subsection{Riemannian Trust Region Method}
To implement the trust region method on Riemannian submanifold, the quadratic model is indispensable, which is the approximation of $f$ on a neighborhood of ${{\bP}}$ defined on $T_{{{\bP}}}\cM$ \cite{RTR1}. Typically, the basic choice of a quadratic model in the $k$-th iteration is
\begin{equation}
	\begin{aligned}
		m_{{{\bP}}^k}\left(\bxi_{{{\bP}}^k}\right)=&f\left({{\bP}}^k\right)+g^{\cM}_{{{\bP}}^k}\left(\mathrm{grad}f\left({{\bP}}^k\right),\bxi_{{{\bP}}^k}\right)\\
		&+\frac{1}{2}g^{\cM}_{{{\bP}}^k}\left(\mathrm{Hess}f\left({{\bP}}^k\right)\left[\bxi_{{{\bP}}^k}\right],\bxi_{{{\bP}}^k}\right),
	\end{aligned}
\end{equation}
where $\mathrm{Hess}f\left({{\bP}}^k\right)\left[\bxi_{{{\bP}}^k}\right]$ is used to form a quadratic model  around the current iteration $\bP^k$ and is not necessarily positive semidefinite. $\mathrm{Hess}f\left({{\bP}}^k\right)\left[\bxi_{{{\bP}}^k}\right]$ can be obtained from \eqref{hessian_TPC}, \eqref{hessian_PUPC} and \eqref{hessian_PAPC}.  
The search direction and the step length are obtained simultaneously by solving the trust-region subproblem 

\begin{equation}\label{subproblem}
	\begin{aligned}
		&\min_{{{\boldsymbol{\eta}}}^k\in T_{{{\bP}}^k}\cM}m_{{{\bP}}^k}\left( {{\boldsymbol{\eta}}}^k \right) \ \ \mathrm{ s.t. }\ g^{\cM}_{\bP^k}\left( {{\boldsymbol{\eta}}}^k,{{\boldsymbol{\eta}}}^k \right)\leq \left(\delta^k\right)^2 ,
	\end{aligned}
\end{equation}
where $\delta^k$ is the trust-region radius in the $k$-th iteration. \eqref{subproblem} can be solved by truncated conjugate gradient (tCG) method  \cite{Conn,Absil2009}, where the Riemannian Hessian needs to be computed according to \eqref{hessian_TPC}, \eqref{hessian_PUPC} or \eqref{hessian_PAPC}. The most common stopping criteria of tCG is to truncate after a fixed number of iterations, denoted by $N_{\text{sub}}$ \cite{AbsilRTR}. Once ${{\boldsymbol{\eta}}}^k$ is obtained, the quality of the quadratic model $m_{{{\bP}}^k}$ is evaluated by the quotient
\begin{equation}\label{rho}
	\begin{aligned}
		\rho^k=\frac{f\left({{\bP}}^k\right)-f\left({{\bP}}^{k+1}\right)}{m_{\bP^k}\left(\boldsymbol{0}\right)-m_{\bP^k}\left({{\boldsymbol{\eta}}}^k\right)}.
	\end{aligned}
\end{equation}
The convergence analyses of the RTR method are available in \appref{App_Proof_Convergence}.

The RTR method for precoder design is summarized in \alref{RTR}, where we use the superscript pair $(k,d)$ for the iteration of  solving the trust region subproblem to distinguish from $(k,n)$ for the iteration of searching for the step length. $N_{\text{sub}}$ is the maximum inner iteration number. In the Step 7 of  \alref{RTR}, $t$ is found by computing the positive root of the quadratic equation \cite{Absil2009}
\begin{equation}
	\begin{aligned}
		t^2g_{\bP^k}^{\cM}\left( {{\bQ}}^{k,d-1},{{\bQ}}^{k,d-1} \right)+2tg_{\bP^k}^{\cM}\left( {{\boldsymbol{\eta}}}^{k,d-1}, {{\bQ}}^{k,d-1}\right)\\
		=\left(\delta^k\right)^2-g_{\bP^k}^{\cM}\left( {{\boldsymbol{\eta}}}^{k,d-1},{{\boldsymbol{\eta}}}^{k,d-1} \right).
	\end{aligned}
\end{equation}
\begin{algorithm}[h]  
	\caption{RTR method for precoder design}  
	\label{RTR}  
	\renewcommand{\algorithmicensure}{\textbf{Output:}}
	\begin{algorithmic}[1] 
		\REQUIRE Riemannian submanifold $\cM$; Riemannian metric $g_{{{\bP}}}\left(\cdot\right)$; Real-valued function $f$; \\
		Retraction $R_{{{\bP}}}\left(\cdot\right)$; initial region $\delta_0\in(0,\Delta]$; threshold $\rho^{\prime}$; inner iteration number $N_{\text{sub}}$.
		\renewcommand{\algorithmicrequire}{\textbf{Input:}}
		\REQUIRE Initial point ${{\bP}}^0$, ${{\bQ}}^{1,0}={{\bG}}^{1,0}=-\mathrm{grad}f({{\bP}}^{0})$, ${{\boldsymbol{\eta}}}^{1,0}=\boldsymbol{0}$, 
		\REPEAT 
		\FOR{$d=1,2,\cdots,N_{\text{sub}}$}
		\STATE  Compute $\mathrm{Hess}f({{\bP}}^{k})\left[{{\bQ}}^{k,d-1}\right]$ with  \eqref{hessian_TPC}, \eqref{hessian_PUPC} or \eqref{hessian_PAPC}.
		\STATE $\tau^{k,d}=\frac{g^{\cM}_{{{\bP}}^{k}}\left( {{\bG}}^{k,d-1},{{\bG}}^{k,d-1} \right)}{g^{\cM}_{{{\bP}}^{k,d-1}}\left({{\bQ}}^{k,d-1},\mathrm{Hess}f({{\bP}}^{k})\left[{{\bQ}}^{k,d-1}\right]\right)}$.
		\STATE  ${{\boldsymbol{\eta}}}^{k,d}={{\boldsymbol{\eta}}}^{k,d-1}+\tau^{k,d-1}{{\bQ}}^{k,d-1}$.
		\IF {$\tau^{k,d}\leq 0$ or $g^{\cM}_{{{\bP}}^k}\left({{\boldsymbol{\eta}}}^{k,d},{{\boldsymbol{\eta}}}^{k,d}\right)\geq \delta^k$}
		\STATE ${{\boldsymbol{\eta}}}^{k,d}={{\boldsymbol{\eta}}}^{k,d-1}+t{{\bQ}}^{k,d-1}$ with $t>0$ and $g({{\boldsymbol{\eta}}}^{k,d},{{\boldsymbol{\eta}}}^{k,d})=\left(\delta^k\right)^2$, $\textbf{break}$.
		\ENDIF
		\STATE ${{\bG}}^{k,d}={{\bG}}^{k,d-1}-\tau^{k,d}\mathrm{Hess}f({{\bP}}^{k})\left[{{\bQ}}^{k,d-1}\right]$.
		\STATE ${{\bQ}}^{k,d}={{\bG}}^{k,d}+\frac{g^{\cM}_{{{\bP}}^{k}}\left({{\bG}}^{k,d},{{\bG}}^{k,d}\right)}{g^{\cM}_{{{\bP}}^{k}}\left({{\bG}}^{k,d-1},{{\bG}}^{k,d-1}\right)}{{\bQ}}^{k,d-1}$.
		\ENDFOR
		\STATE ${{\boldsymbol{\eta}}}^{k}={{\boldsymbol{\eta}}}^{k,d}$, $\nu =g^{\cM}_{{{\bP}}^k}\left({{\boldsymbol{\eta}}}^k,{{\boldsymbol{\eta}}}^k\right)$. 
		\STATE Evaluate $\rho^k$ with \eqref{rho}.
		\STATE Update ${{\bP}}^{k+1}$ and $\delta^{k+1}$ according to $\rho^k$ with 
		\begin{equation}
			\begin{aligned}			
			&{{\bP}}^{k+1}=
			\begin{cases}
				\eqref{Retraction}, \ \text{if}\  \rho^k>\rho^{\prime} \\
				{{\bP}}^k, \ \text{if}\ \rho^k\leq\rho^{\prime}
			\end{cases}\notag,\\
			&\delta^{k+1} =
			\begin{cases}
				\frac{1}{4}\delta^k, \text{if}\ \rho^k<\frac{1}{4}\\
				\min\left(2\delta^k,\Delta\right), \text{if}\ \rho^k>\frac{3}{4}\ \& \ \nu=\delta^k\\
				\delta^k,  \text{otherwise} 
			\end{cases}.
		\end{aligned}
		\end{equation}
	   \STATE Set $k\leftarrow k+1$.
		\UNTIL{convergence}
	\end{algorithmic}  
\end{algorithm}

\subsection{Computational Complexity of Riemannian Methods}
\renewcommand\thesubsectiondis{\thesection.\arabic{subsection}}
Let $N_r=\sum_{i=1}^{U}M_i$ and $N_d=\sum_{i=1}^{U}d_i$ denote the total antenna number of the users and the total data stream number, respectively. Let $N_{\text{out}}$ and $N_{\text{in}}$ denote the numbers of outer iteration and inner iteration of searching for the step length, respectively. The computational costs of the Riemannian metric, orthogonal  projection, retraction and vector transport are the same, i.e., $M_tN_d$. Although they need to be called several times per outer iteration in the proposed Riemannian methods, the total computational costs can still be neglected as they are much lower than that of the Riemannian gradient or Riemannian Hessian.
For RSD and RCG methods on $\widehat{\cM}$ and $\widetilde{\cM}$, ${{\bV}}_{i,\ell}^k$ and ${{\bV}}_{i,\ell}^{k,n},\forall i,\ell\in\cU$, in \alref{RCG} can be obtained directly from the elements computed in the previous iteration with  \eqref{V_1}, \eqref{V_2}, \eqref{V_inner_1} and \eqref{V_inner_2}. Moreover, the dimension of the input of the cost function $\logdet{\cdot}$ is $M_i$ for the $i$-th UT. Thus, the computational cost of the whole inner iteration per outer iteration is $O\left(2N_{\text{in}}\sum_{i=1}^UM_i^3\right)$, which is very low and can be omitted.  So the computational cost mainly comes from Step 3 and Step 5 in \alref{RCG} during the iterations, and the complexity is $O(M_tN_rN_d+M_tN_d^2)N_{\text{out}}$. For RSD and RCG methods on $\overline{\cM}$,  $\bar{\bV}_{i,\ell}^{k,n},\forall i,\ell\in\cU$, have to be computed once for each inner iteration according to \eqref{V_inner_3}, with the total complexity of $O((N_{\text{in}}+1)M_tN_rN_d+M_tN_d^2)N_{\text{out}}$. The main computational cost of the RTR method depends on Step 3 in \alref{RTR}, where the Riemannian Hessian needs to be computed once in each inner iteration. The computational cost of the Riemannian Hessian mainly comes from the matrix multiplication  according to \appref{App_Proof_Theo_TPC}, \appref{App_Proof_Theo_Riemannian_Hessian_sub_PUPC} and \appref{App_Proof_Theo_Riemannian_Hessian_sub_PAPC}. The maximum complexity of the RTR method is $O(8N_{\text{sub}}M_tN_rN_d)N_{\text{out}}$ on $\widehat{\cM}$ or $\overline{\cM}$ and $O(4N_{\text{sub}}M_tN_rN_d)N_{\text{out}}$ on $\widetilde{\cM}$.
 Typically,  we have $N_{\text{in}} \leq 10$ with $r=0.5$, $c=10^{-4}$ and $\alpha^0=10^{-3}$, and $N_{\text{sub}}\leq 10$.

The computational complexities of different Riemannian design methods are summarized in \tabref{ComplexityOfRiemannianMethods}. For the same power constraint, we can see that RSD and RCG methods have the same computational complexity, which is lower than that of the RTR method.
\renewcommand\arraystretch{1.2}
\begin{table*}
	\centering
	\caption{Computational complexity of different Riemannian  design methods} 
	\label{ComplexityOfRiemannianMethods}
	\begin{tabular}{|c|c|c|c|}
		\hline
		&RSD&RCG&RTR\\
		\hline		
		TPC&$O(M_tN_rN_d+M_tN_d^2)N_{\text{out}}$ &$O(M_tN_rN_d+M_tN_d^2)N_{\text{out}}$ &$O(8N_{\text{sub}}M_tN_rN_d)N_{\text{out}}$\\
		\hline
		PUPC&$O(M_tN_rN_d+M_tN_d^2)N_{\text{out}}$ &$O(M_tN_rN_d+M_tN_d^2)N_{\text{out}}$ &$O(4N_{\text{sub}}M_tN_rN_d)N_{\text{out}}$ \\
		\hline 
		PAPC&$O((N_{\text{in}}+1)M_tN_rN_d+M_tN_d^2)N_{\text{out}}$ &$O((N_{\text{in}}+1)M_tN_rN_d+M_tN_d^2)N_{\text{out}}$ &$O(8N_{\text{sub}}M_tN_rN_d)N_{\text{out}}$\\
		\hline 
	\end{tabular}
\end{table*}

The popular WMMSE method proposed in \cite{WMMSE} has the same objective function as our proposed Riemannian methods for precoder design.
The computational complexity of the  WMMSE method   is $O(2M_tN_rN_d+M_tN_d^2+\frac{3}{2}N_d^3)N_{\text{out}}$, which is higher than that of RCG method under TPC, in the case with the same number of outer iterations used. SLNR precoder also satisfies PUPC, where the eigenvalue decomposition needs to be performed $U$ times \cite{SLNREE}. Thus, the  computational complexity of the SLNR precoder is $O\left(M_t^3U\right)$, which is higher than the RCG method when $N_{\text{out}}< \frac{M_t^2U}{N_rN_d+N_d^2}$. In addition, the precoder under PUPC can also be designed by the WMMSE method, where $U$ Lagrange multipliers are obtained by the bisection method \cite{WMMSE_shi}. The complexity of WMMSE under PUPC is $O(M_t^3U+2M_tN_rN_d+M_tN_d^2+\frac{3}{2}N_d^3)N_{\text{out}}$, which is much higher than that of RCG method.  WSR-maximization precoders under PAPC are difficult to design. \cite{ZFBased} provides an alternative method by finding the linear precoder closest to the optimal ZF-based precoder under TPC, whose complexity is $O(M_t^3+(M_t-N_r)^3U)N_{\text{Z}}$ and $N_{\text{Z}}$ is the number of iterations needed. Despite no inner iteration, this method has a higher computational complexity than that of the RCG method in massive MIMO systems when $N_{\text{Z}}=N_{\text{out}}$. \cite{DCAM} designs the WSR-maximization precoder under PAPC via the DCAM-based method,  whose computational complexity is $O\left(N_{\text{D}}\left(M_t^3+M_t^2N_d\right)\right)N_{\text{out}}$ with $N_{\text{D}}$ being the number of iterations needed for the DCAM method to converge. Typically,  the optimality tolerance of the  duality gap is $\epsilon=10^{-4}$ in \cite{DCAM} and $N_{\text{D}}\propto \logtwo{\frac{1}{\epsilon}}$, indicating that the computational complexity of  the DCAM-based method is higher than that of the RCG design method per outer iteration.
The numerical results in the next section will demonstrate that Riemannian design methods converge faster than the comparable iterative methods under different power constraints.

 \section{Numerical Results}\label{Results}
 
 In this section, we provide simulation results to demonstrate that the precoders designed in the matrix manifold framework are numerically superior and computationally efficient under different power constraints, especially when the RCG method is used. We adopt the prevalent QuaDRiGa channel model \cite{QuaDRiGa}, where ``3GPP 38.901 UMa NLOS'' scenario is considered. The 25 m high BS with $M_t=128$ antennas serves $U=20$ 1.5 m high UTs, which are randomly distributed in the cell with the radius of 250 m. The antenna type is ``3gpp-3d'' at the BS side and ``ula'' at the user side. For simplicity, we set $M_i=d_i=2, P_i=\frac{P}{U}, \forall i\in\cU$. The center frequency  is set at 4.8 GHz. The weighted factor of each user is set to $1$. The powers of the channels for different users are normalized. The SNR is defined as $\frac{P}{\sigma_z^2}$, where the noise power $\sigma_z^2$ is set to be 1 and different SNRs are obtained by adjusting the transmit power $P$. For fair comparison, the Riemannian methods and the iterative methods for comparison are all initialized by the regularized zero-forcing (RZF) precoders \cite{RZF}, which are forced to satisfy the TPC, PUPC and PAPC, respectively. 
\begin{figure}[b]
	\begin{minipage}[htbp]{0.5\textwidth}
		\centering
		\includegraphics[scale=0.42]{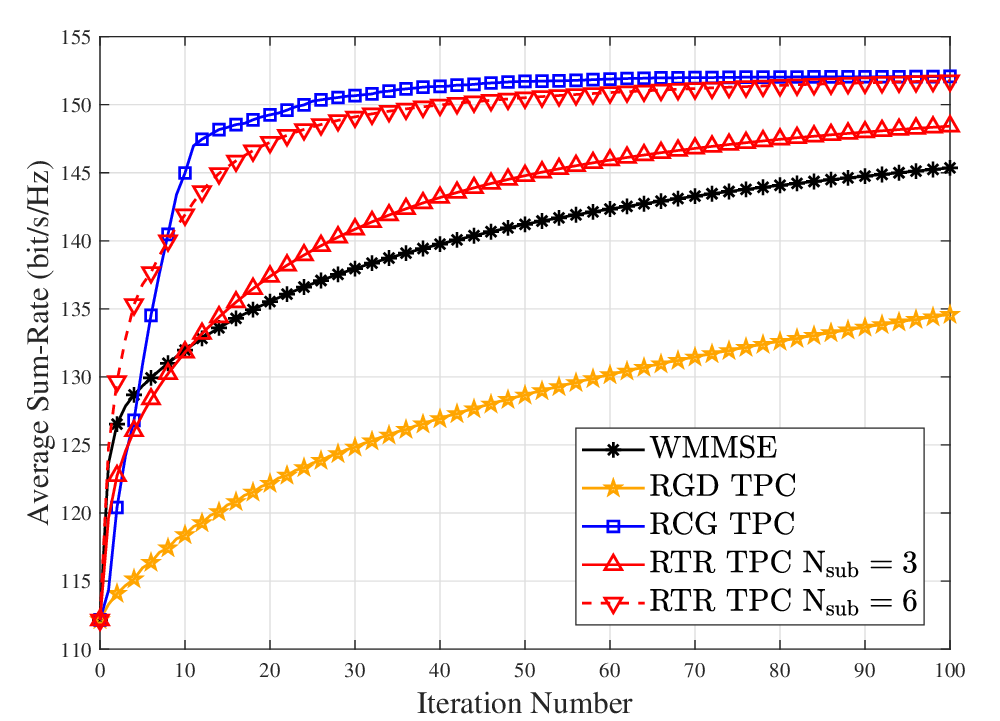}
		\caption{Convergence comparison under TPC at SNR=20dB.}
		\label{fig_conTPC}
	\end{minipage}
\end{figure}
 
\begin{figure}[t]
	\begin{minipage}[htbp]{0.5\textwidth}
		\centering
		\includegraphics[scale=0.42]{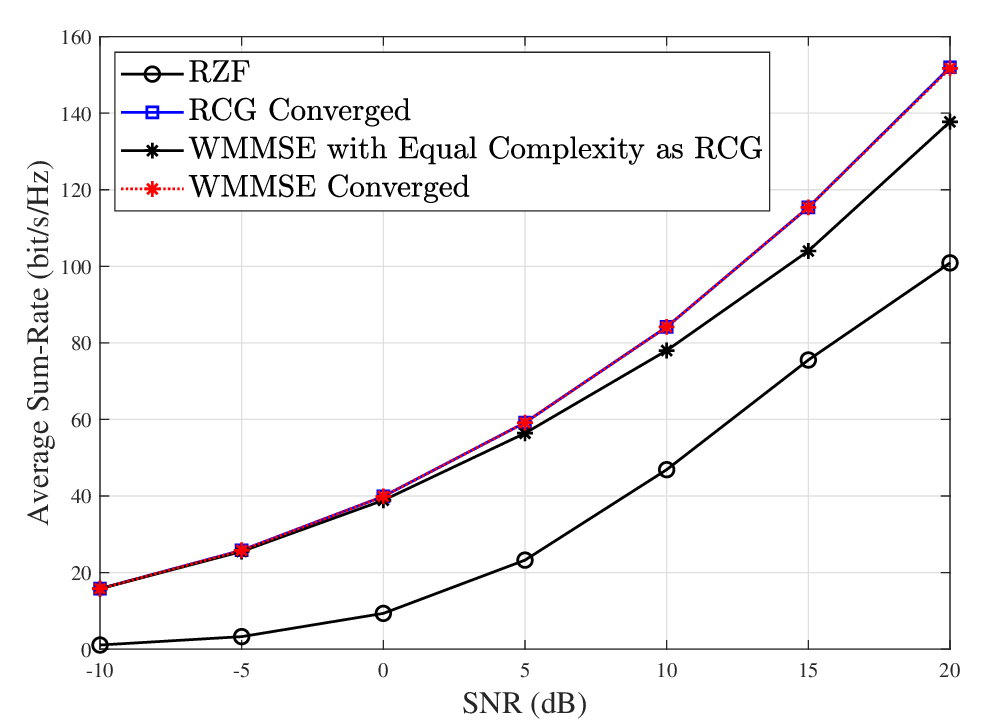}
		\caption{WSR performance when the RCG method converges under TPC.}
		\label{fig_TPC}
	\end{minipage}
	\begin{minipage}[t]{0.5\textwidth}
			\centering
			\includegraphics[scale=0.42]{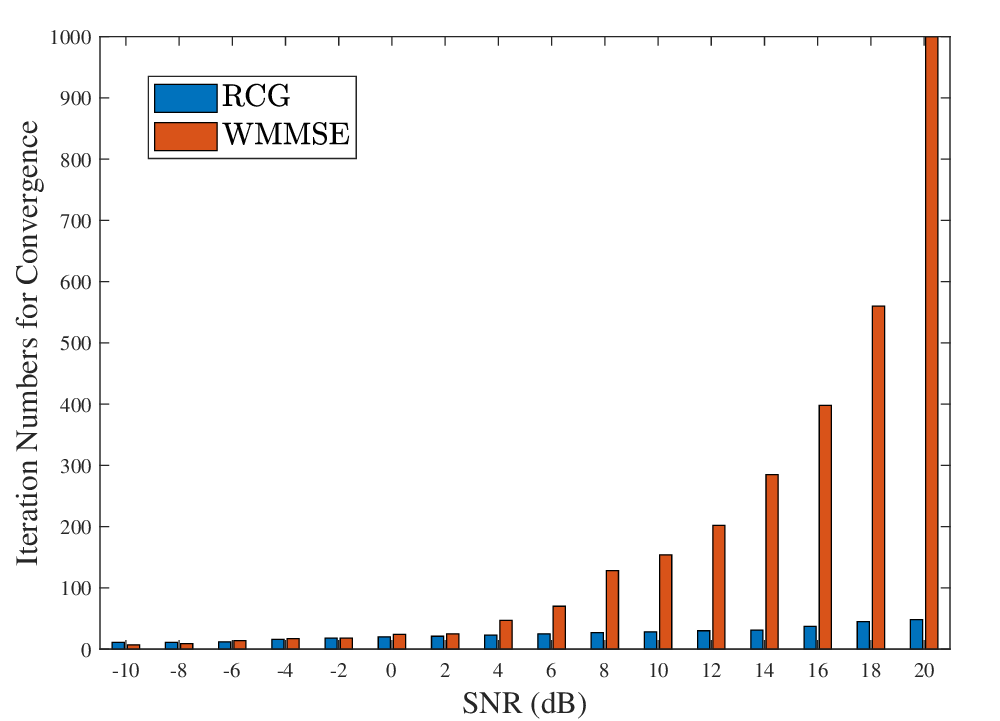}
			\caption{Convergence comparison between RCG and  WMMSE under TPC.}
			\label{fig_RCGvsWMMSE}
		\end{minipage}
\end{figure}
\figref{fig_conTPC} depicts the convergence trajectories of the Riemannian design methods compared with the WMMSE method \cite{WMMSE} under TPC at SNR$=20$  dB for massive MIMO DL. As shown in \figref{fig_conTPC}, the RCG method converges much faster than RSD, WMMSE and RTR methods when $N_{\text{sub}}=3$. RTR method converges the fastest at the very beginning  when $N_{\text{sub}}=6$ and obtains $87\%$ performance in the first three iterations. Besides, the complexity of the RCG method is lower than that of the WMMSE method.  For the fairness of comparison, \figref{fig_TPC} compares the WSR performance of the RCG method with that of the WMMSE precoder under the same complexity when the RCG method converges. As we can see from \figref{fig_TPC}, the RCG method has the same performance as the WMMSE method when they converge and has an evident performance gain in the high SNR regime when they share the same computational complexity. The RZF method \cite{RZF}  requires the least computation among these methods but also exhibits the poorest performance. To investigate the convergence behavior of the RCG method, we compare the approximate iterations needed for the RCG and WMMSE methods to completely converge in \figref{fig_RCGvsWMMSE}. Compared with the WMMSE method, the RCG method needs much fewer iterations to converge, especially when SNR is high, which clearly demonstrates the fast convergence and high computational efficiency of the RCG method.

\begin{figure}[t]

\begin{minipage}[t]{0.5\textwidth}
	\centering
	\includegraphics[scale=0.42]{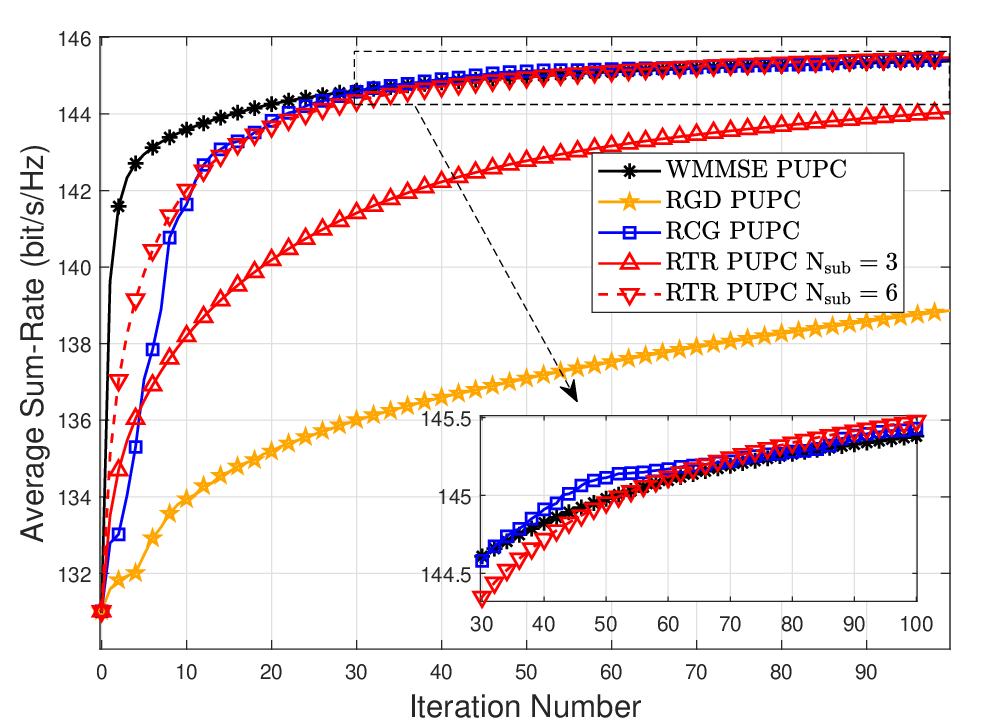}
	\caption{Convergence comparison under PUPC at SNR=20dB}
	\label{fig_conPUPC}
\end{minipage}

	\begin{minipage}[t]{0.5\textwidth}
		\centering
		\includegraphics[scale=0.42]{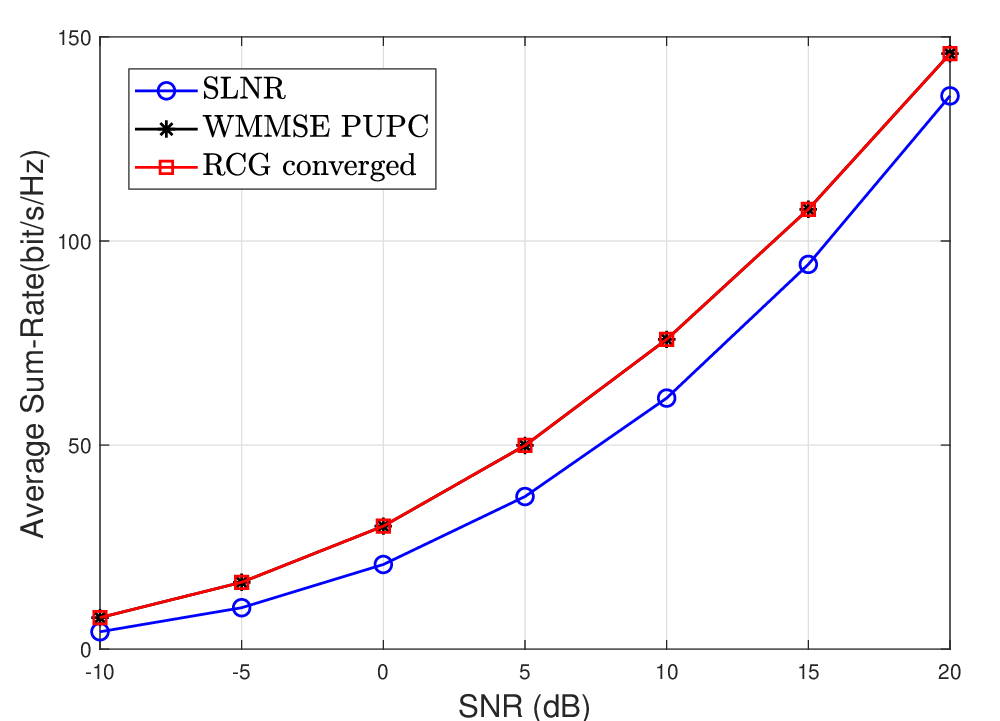}
		\caption{WSR performance when the RCG method converges under PUPC.}
		\label{fig_PUPC}
	\end{minipage}
\end{figure}

 \begin{figure}[htbp]
	\begin{minipage}[t]{0.5\textwidth}
		\centering
		\includegraphics[scale=0.42]{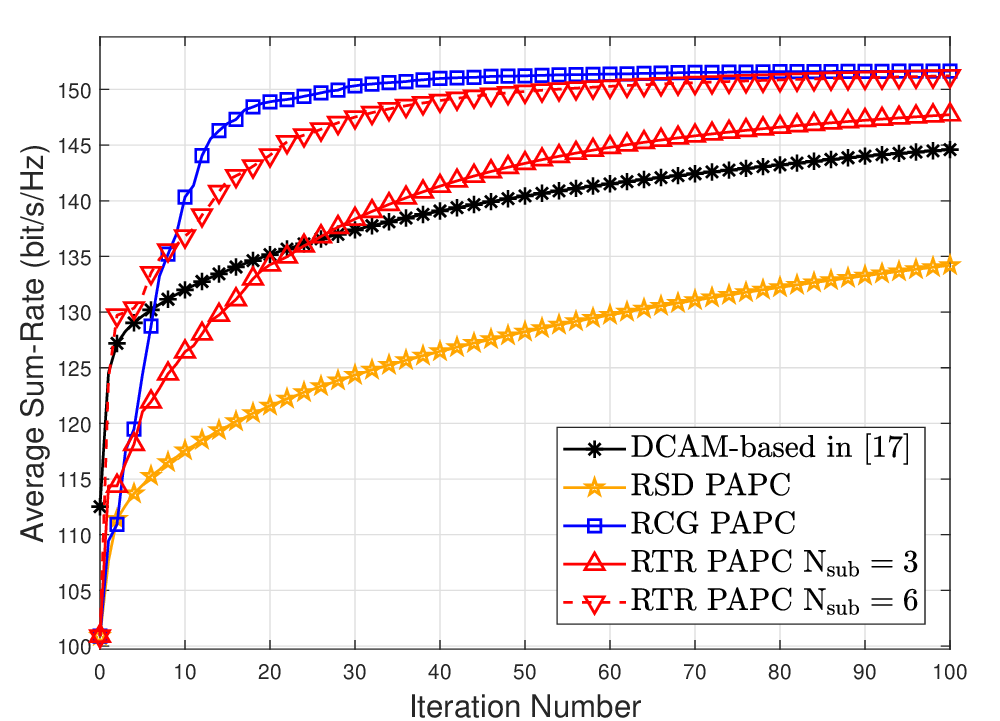}
		\caption{Convergence comparison under PAPC at SNR=20dB.}
		\label{fig_conPAPC}
	\end{minipage}
\end{figure}
 \begin{figure}[htbp]
	\begin{minipage}[t]{0.5\textwidth}
		\centering
		\includegraphics[scale=0.42]{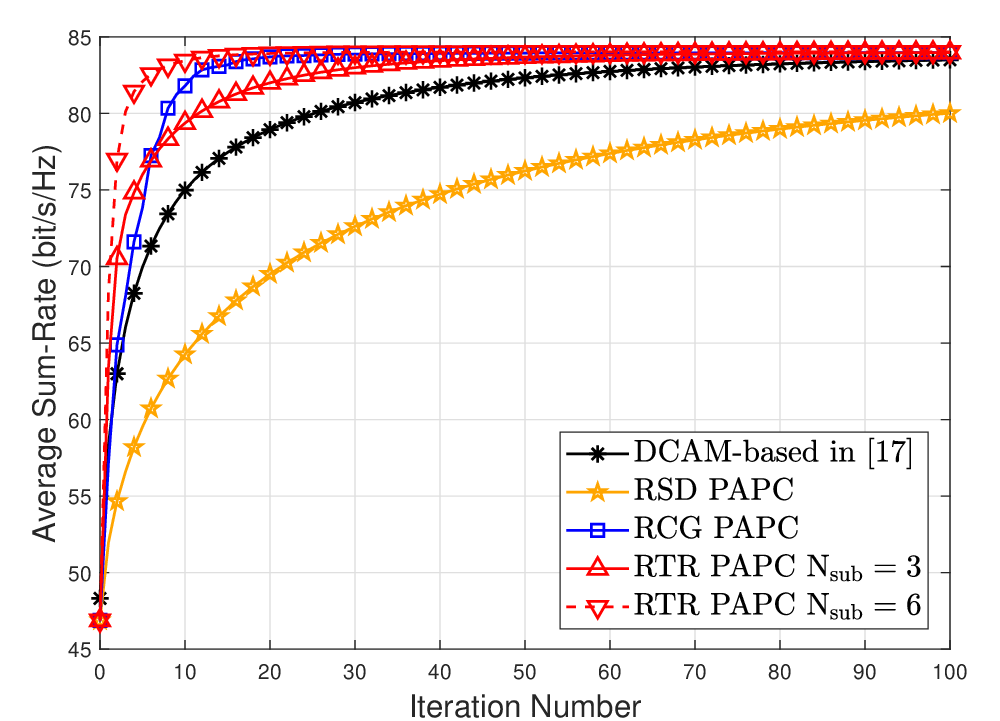}
		\caption{Convergence comparison under PAPC at SNR=10dB.}
		\label{fig_conPAPC10}
	\end{minipage}
\end{figure}
\begin{figure}[htbp]
	\begin{minipage}[tp]{0.5\textwidth}
		\centering
		\includegraphics[scale=0.42]{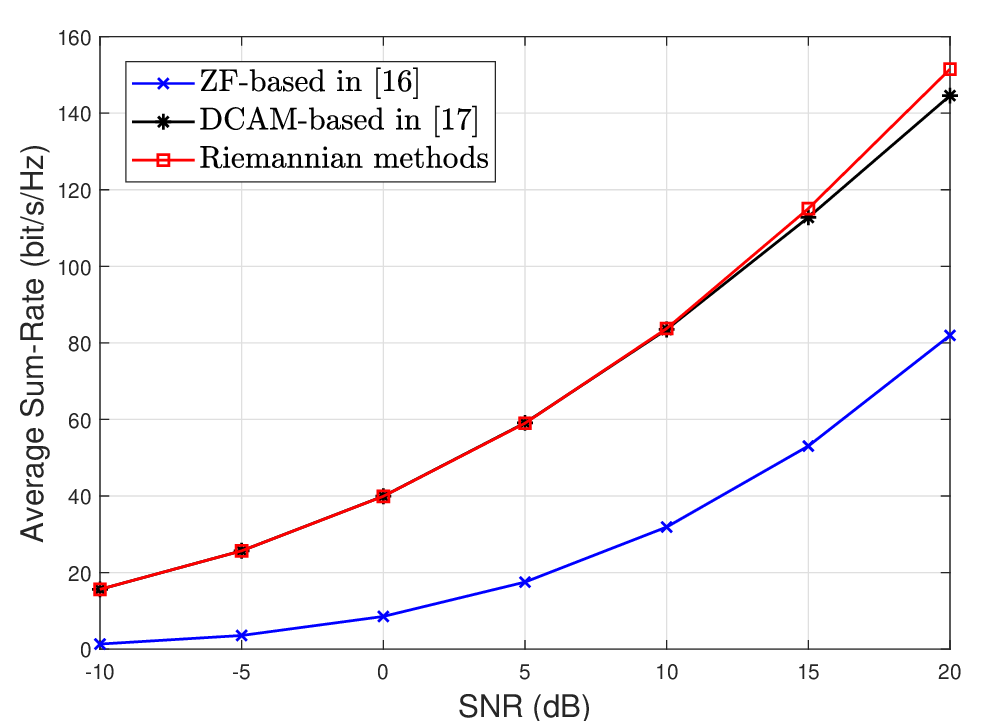}
		\caption{WSR performance when the RCG method converges under PAPC.}
		\label{fig_PAPC}
	\end{minipage}
\end{figure}
The convergence behavior of Riemannian design methods under PUPC at SNR$=20$ dB is shown in \figref{fig_conPUPC}. When $N_{\text{sub}}=6$, the RTR method needs nearly the same number of iterations to converge as RCG at the cost of a much higher computational complexity.  In addition, WMMSE achieves good performance and converges fast at the first thirty iterations. However, RCG shows a better performance and converges faster compared with WMMSE when $N_{\text{out}}>30$ with a much lower computational complexity.
In \figref{fig_PUPC}, we compare the WSR performance of the Riemannian methods after convergence under different SNRs with those of the SLNR precoder and the WMMSE precoder satisfying PUPC, whose computational complexities are much higher than that of the RCG method. The WSR performances of the proposed Riemannian methods are almost the same after convergence with different initializations.  From \figref{fig_PUPC}, the RCG method has the same performance as the  WMMSE method after convergence and performs better than the SLNR precoder in the whole SNR regime.
 
\figref{fig_conPAPC} shows the convergence behavior of Riemannian design methods compared with the DCAM-based method in \cite{DCAM} under PAPC at SNR$=20$ dB. From \figref{fig_conPAPC}, we find that the RTR method converges fast in the first three iterations, but the RCG method converges faster when $N_{\text{out}}>10$. RCG and RTR methods both converge much faster than the DCAM-based method, whose complexity is higher than that of the RCG design method per outer iteration. To further show the convergence behavior of the Riemannian methods in the PAPC case under lower SNR, we plot the convergence trajectories of the Riemannian methods under PAPC at SNR=10 dB in \figref{fig_conPAPC10}. From \figref{fig_conPAPC10}, we can see that the  RTR with $N_{\text{sub}}=6$ converges fastest and has the same WSR performance as the RCG method after convergence.  \figref{fig_PAPC} compares the WSR performance of the RCG method after convergence with the ZF-based method in \cite{ZFBased} and the DCAM-based method in \cite{DCAM}.  As we can see in \figref{fig_PAPC}, the RCG method is numerically superior in designing precoders under PAPC with a lower computational complexity compared with the algorithm proposed in \cite{ZFBased}. The DCAM-based method  exhibits similar performance with the RCG method in the low SNR regime at a higher computational cost but  performs worse than the RCG method when SNR is high. In addition to the RZF precoder, we also try the random initialization and the WSR performances of the Riemannian methods after convergence are nearly the same.
Finally, we compare the running times required for the proposed Riemannian methods and the  comparable iterative methods to converge in \tabref{CPUtime}. All the simulations are conducted by Matlab R2019b on a desktop with Intel(R) Core(TM) i9-10900K running at 3.70 GHz. We can see that the RCG method runs fastest in all the cases. In addition, RTR  runs faster than the iterative methods for comparison in all the cases, especially in the PUPC case.
\renewcommand\arraystretch{1.5}
\begin{table*}
	\centering
	\caption{Running times of different methods (s)}
	\label{CPUtime}
	\small
	\begin{tabular}{c|c|c|c|c}
		\hline
		
		 &RSD &RCG &RTR &Iterative methods\\
         \hline
		 
		TPC&100.5072&4.0844 &44.8081 &\ \ \ WMMSE-TPC: 73.6469 \\
		\rule{0pt}{12pt} 
		PUPC  &90.2711&8.9171 &17.3877 & WMMSE-PUPC: 29.8341\\
		\rule{0pt}{12pt} 
		 PAPC &175.0686 &16.5449 & 77.3072  &\hspace{.52cm}   DCAM-based: 145.9009 \\
		\hline

	\end{tabular}
\end{table*}

%
 \section{Conclusion}\label{Conclusion}
In this paper, we have investigated the linear precoder design methods with matrix manifold in massive MIMO DL transmission.  We focus on the WSR-maximization precoder design and demonstrate that the precoders under TPC, PUPC and PAPC are on different Riemannian submanifolds. Then, the constrained problems are transformed into unconstrained ones on Riemannian submanifolds. Furthermore, RSD, RCG and RTR methods are proposed for optimizing on Riemannian submanifolds. There is no inverse of large dimensional matrix during the iterations in the proposed methods. Besides, the complexities of implementing these Riemannian design methods on different Riemannian submanifolds are investigated. Simulation results show the numerical superiority and computational efficiency of the RCG method.
\appendices
\renewcommand\thesubsectiondis{\thesection.\arabic{subsection}}
\section{Proof of \thref{theo_submanifold_TPC}}\label{Proof_theo_Riemannian_Submanifold}
\subsection{Proof in the TPC case}\label{App_proof_Theo_submanifold_TPC}
With the fact that the Frobenius norm of $\hat{\bP}$ is a constant, the constraint $\mathrm{tr}\left(\bP^H\bP\right)=P$ defines a sphere naturally. Clearly, $ \widehat{\cM} $ is a subset of the set $  \bbC^{M_t \times d_1} \times \bbC^{M_t \times d_2}\times \cdots \times \bbC^{M_t \times d_i}$. Consider  differentiable function $ \hat{F}  : \cN \rightarrow \bbR : \bP \mapsto \mathrm{tr}\left(\bP\bP^H\right)-P $, and clearly $  \widehat{\cM} \in \hat{F}^{-1}\left( 0\right)  $, where $ 0 \in \bbR $. Based on the submersion theorem \cite[Proposition 3.3.3]{Absil2009}, to show $ \widehat{\cM } $ is a closed embedded submanifold of $\cN $, we need to prove $ \hat{F} $ as a submersion at each point of $\widehat{\cM }  $. In other words, we should verify that the rank of $ \hat{F} $ is equal to the dimension of $ \bbR $, i.e., 1, at every point of $\widehat{\cM}  $. Since the rank of $\hat{F} $ at a point $ \bP \in \cN $ is defined as the dimension of the range of $ \mathrm{D}\hat{F}\left( \bP\right)  $, we simply need to show that for all $ w\in\bbR $,  there exists $ \bZ = \left( \bZ_1,\bZ_2,\cdots,\bZ_i \right)   \in \cN $ such that $ \mathrm{D}\hat{F}\left( \bP\right)\left[ \bZ\right] = w $. Since the differential operation at  $ \bP $ is equivalent to the component-wise differential at  each of $ \bP_{1},\bP_{2},\cdots, \bP_i$, we have 
{\setlength\abovedisplayskip{0.1cm}
	\setlength\belowdisplayskip{0.1cm}
\begin{equation}\label{key}
	\mathrm{D}\hat{F}\left( \bP\right)\left[ \bZ\right] = \mathrm{tr}\left(\bP^H\bZ+\bZ^H\bP\right).
\end{equation}}
It is easy to see that if we choose $ \bZ_i = \frac{1}{2} \frac{w}{P} \bP_i $, we will have $\mathrm{D}\hat{F}\left( \bP\right)\left[ \bZ\right] = w $. This shows that $ \hat{F} $ is full rank as well as a submersion on $ \widehat{\cM } $. 

\par Because every tangent space $ T_{\hat{\bP}}\widehat{\cM }  $ is a subspace of  $ T_{\bP} \cN $, the Riemannian metric $ g_{\bP}\left(\cdot\right) $ of $  \cN $ is naturally introduced to $ \widehat{\cM } $ as $ g_{\hat{\bP}}\left(\cdot\right) = g_{\bP}\left(\cdot\right)$. With this metric, $ \widehat{\cM } $ is an Riemannian embedded submanifold of $ \cN $.
\subsection{Proof in the PUPC case}\label{App_proof_theo_submanifold_PUPC}
Note that each $\tilde{\bP}_i$ in $\tilde{\bP}$ forms a sphere. Thus $\widetilde{\cM}$ forms an oblique manifold composed of $U$ spheres.
To show $\widetilde{\cM}$ is a Riemannian submanifold of the product manifold $\cN$, we need to show that for  $ \bD_i=a_i\bI_{d_i}, \forall i\in \cU$, where $a_i\in \bbR$, there exists $ \bZ = \left( \bZ_1,\cdots,  \bZ_U \right)   \in \cN $ such that $ \mathrm{D}\tilde{F}\left( \bP\right)\left[ \bZ\right] = \mathrm{blkdiag}\left( \bD_1,\cdots,\bD_U \right)=\bD $. Since the differential operation at  $ \bP $ is equivalent to the component-wise differential at each of $ \bP_{1},\bP_{2},\cdots, \bP_U$, we have 
\begin{equation}\label{key}
	\begin{aligned}
		&\mathrm{D}\tilde{F}\left( \bP\right)\left[ \bZ\right] =\\
		&\left(\mathrm{tr}\left(\bP_1^H\bZ_1+\bZ_1^H\bP_1\right)\bI_{d_1},\cdots,\mathrm{tr}\left(\bP_U^H\bZ_U+\bZ_U^H\bP_U\right)\bI_{d_U}\right).
	\end{aligned}
\end{equation}
It is easy to see that if we choose $ \bZ_i = \frac{1}{2} \frac{U}{P} \bP_i\bD_i $, we will have $\mathrm{D}\tilde{F}\left( \bP\right)\left[ \bZ\right] = \bD $. This shows that $ \tilde{F} $ is full rank as well as a submersion on $ \widetilde{\cM } $. 

\par Because every tangent space $ T_{\tilde{\bP}}\widetilde{\cM }  $ is a subspace of  $ T_{\bP} \cN $, the Riemannian metric $ g_{\bP}\left(\cdot\right) $ of $  \cN $ is naturally introduced to $ \widetilde{\cM } $ as $ g_{\tilde{\bP}}\left(\cdot\right)= g_{\bP}\left(\cdot\right) $. With this metric, $ \widetilde{\cM } $ is an Riemannian embedded submanifold of $ \cN $.
\subsection{Proof in the PAPC case}\label{App_proof_Theo_submanifold_PAPC}
Every column of $\bar{\bP}^H$ forms a sphere and thus $\bar{\bP}^H$ forms a standard oblique manifold composed of $M_t$ spheres.
Similarly, to show $\overline{\cM}$ is a Riemannian submanifold of the product manifold $\cN$, we need to show that for all $ \bw = \left( w_1, w_2, \cdots, w_{M_t} \right)   \in  \bbR^{M_t} $, there exists $ \bZ = \left( \bZ_1,\bZ_2,\cdots,\bZ_U \right)   \in \cN $ such that $ \mathrm{D}F\left( \bP\right)\left[ \bZ\right] = \diag{\bw} $. Since the differential operation at $ \bP $ is equivalent to the component-wise differential at each of $ \bP_{1},\bP_{2},\cdots, \bP_i$, we have 
\begin{equation}\label{key}
	\mathrm{D}\bar{F}\left( \bP\right)\left[ \bZ\right] = \bI_{M_t} \odot \sum_{i=1}^{U}\left( \bP_i \bZ_i^H + \bZ_i \bP_i^H\right).
\end{equation}
It is easy to see that if we choose $ \bZ_i = \frac{1}{2} \frac{M_t}{P} \diag{\bw} \bP_i $, we will have $\mathrm{D}\bar{F}\left( \bP\right)\left[ \bZ\right] = \diag{\bw} $. This shows that $ \bar{F} $ is full rank as well as a submersion on $ \overline{\cM } $. 

\par Because every tangent space $ T_{\bar{\bP}}\overline{\cM }  $ is a subspace of  $ T_{\bP} \cN $, the Riemannian metric $ g_{\bP}\left(\cdot\right) $ of $  \cN $ is naturally introduced to $ \overline{\cM } $ as $ g_{\bar{\bP}}\left(\cdot\right) = g_{\bP}\left(\cdot\right) $. With this metric, $ \overline{\cM } $ is a Riemannian embedded submanifold of $ \cN $.

\section{Proof of \thref{Theo_Euclidean_Gradient}} \label{App_Proof_Theo_Euclidean_Gradient}
To simplify the notations, we use $ f \left( \bP_i\right)  $ to represent $ f \left( \bP\right) $ that only considers $ \bP_i $ as variable with $ \bP_{\ell} $ for $ \ell \neq i $ fixed. 
\par For any $ \bxi_{\bP_i} \in T_{\bP_i} \bbC^{M_t \times M_i} $, the directional derivative of $  f \left( \bP_i\right)  $ along $ \bxi_{\bP_i} $ is  
\begin{equation}\label{key}
	\mathrm{D} f \left( \bP_i\right)\left[ \bxi_{\bP_i}\right]   = -w_i \mathrm{D}\cR_{i}\left[ \bxi_{\bP_i}\right] - \sum_{ \ell \neq i}^U w_{\ell} \mathrm{D}\cR_{\ell}\left[ \bxi_{\bP_i}\right].
\end{equation}
We derive $ \mathrm{D}\cR_{i}\left[ \bxi_{\bP_i}\right] $ and $  \mathrm{D}\cR_{\ell}\left[ \bxi_{\bP_i}\right] $ separately as  
\begin{align}
	\mathrm{D} \cR_{i} \left[ \bxi_{\bP_i}\right]  =& \mathrm{tr}\left(\bC_{i}  \left( \bxi_{\bP_i}^H\bH_i^H\bA_i+\bA_i^H\bH_i\bxi_{\bP_i}\right)\right),
\end{align} 
\begin{align}
	\mathrm{D}\cR_{\ell}\left[ \bxi_{\bP_i}\right] 
	=& - \mathrm{tr}\bigg(\bH_{\ell}^H\bB_{\ell}\bH_{\ell} \left( \bxi_{\bP_i}\bP_i^H + \bP_i \bxi_{\bP_i}^H \right)  \bigg).
\end{align}
Thus, we have
{\setlength\abovedisplayskip{0.1cm}
	\setlength\belowdisplayskip{0.1cm}
	\begin{equation}
		\begin{aligned}
			\mathrm{D} &f\left( \bP_i\right)\left[ \bxi_{\bP_i}\right] 
			= \\
			&g_{\bP_i}\Big(-2  w_i\bH_i^H \bA_i\bC_i + 2\sum_{ \ell \neq i}w_{\ell}\bH_{\ell}^H\bB_{\ell}\bH_{\ell} \bP_i, \bxi_{\bP_i} \Big) 
		\end{aligned}
	\end{equation}}
and $ \mathrm{grad} f \left( \bP_i\right) $ is 
\begin{equation}\label{key}
		\mathrm{grad} f \left( \bP_i\right) = -2\Big( w_i\bH_i^H\bA_i \bC_i - \sum_{\ell\neq i}w_{\ell}\bH_{\ell}^H\bB_{\ell}\bH_{\ell}\bP_i\Big) .
\end{equation}
\section{Proof of \lmref{projection_submanifold_TPC}}\label{App_proof_prop_projection_TPC}
From the decomposition \eqref{Decomposition_TPC} and the definition of the normal space \eqref{Normal_Space_TPC}, it is easy to have $ \Pi_{T_{\hat{\bP}}\widehat{\cM}}^{T_{\bP}\cN}\left( \bxi_{\bP}\right) = \bxi_{\bP} - \hat{\lambda}_1 \bP $. Meanwhile, from the definition of the tangent space \eqref{Tangent_Space_TPC}, $\bxi_{\bP} - \hat{\lambda}_1\bP  $ should satisfy the equation 
\begin{equation}\label{key}
	\mathrm{tr}\left(\bP^H \left( \bxi_{\bP} - \hat{\lambda}_1  \bP\right)  + \left( \bxi_{\bP} - \hat{\lambda}_1  \bP \right)^H  \bP\right) = 0. 
\end{equation}
After some algebra, we can get $\hat{\lambda}_1 =  \frac{1}{P}   \Re\left\lbrace  \mathrm{tr}\left(\bP^H\bxi_{\bP}\right) \right\rbrace $.
\vspace{0em}
\section{Proof of \thref{Theo_Riemannian_Gradient_sub_TPC}}\label{App_Proof_Theo_TPC}
\subsection{Proof for Riemannian gradient in \thref{Theo_Riemannian_Gradient_sub_TPC}}
\par From \cite[Section 3.6.1]{Absil2009} and \lmref{projection_submanifold_TPC}, the Riemannian gradient of $ f \big( \hat{\bP} \big)  $ is 
{\setlength\abovedisplayskip{0.1cm}
	\setlength\belowdisplayskip{0.1cm}\begin{equation}\label{gradient_sub}
	\mathrm{grad} f \big( \hat{\bP} \big)  = \Pi_{T_{\hat{\bP}}\widehat{\cM}}^{T_{\bP}\cN} \left( \mathrm{grad} f\big(\hat{ \bP} \big) \right) 
	= \mathrm{grad} f\left( \bP \right) -\hat{\lambda}_1 \bP,
\end{equation}}
where $ \hat{\lambda}_1 =  \frac{1}{P}  \Re \left\lbrace \mathrm{tr}\left( \bP\big( \mathrm{grad} f \left( \bP\right) \big) ^H\right) \right\rbrace  $.
\subsection{Proof for Riemannian Hessian in \thref{Theo_Riemannian_Gradient_sub_TPC}} \label{App_Proof_Theo_Riemannian_Hessian_sub_TPC}
Combine \eqref{RiemannianConnection} and \eqref{RawComposition}, we get
\begin{equation}\label{hessian_first}
	\begin{aligned}
		\mathrm{Hess}f(\hat{\bP})[\bxi_{\hat{\bP}}]&=\nabla_{\bxi_{\hat{\bP}}}^{\widehat{\cM}}\mathrm{grad}f\\
		&=\Pi_{T_{\hat{\bP}}\widehat{\cM}}^{T_{\bP}\cN}\left( \nabla_{\bxi_{\hat{\bP}}}^{\cN} \mathrm{grad}f\right)\\
		&=\Pi_{T_{\hat{\bP}}\widehat{\cM}}^{T_{\bP}\cN}\left(\mathrm{Dgrad}f(\hat{\bP})[\bxi_{\hat{\bP}}]  \right),
	\end{aligned}
\end{equation}
which belongs to $T_{\hat{\bP}}\widehat{\cM}$. So the following relationship holds:
\begin{equation}\label{hatLambda2}
	\begin{aligned}
		\Re\left\{\mathrm{tr}\left(\bP^H\mathrm{Dgrad}f(\hat{\bP})[\bxi_{\hat{\bP}}]-\hat{\lambda}_2\bP \right)  \right\}=0 \\
		\Rightarrow\hat{\lambda}_2=\frac{1}{P}\Re\left\{\mathrm{tr}\left(\bP^H\mathrm{Dgrad}f(\hat{\bP})[\bxi_{\hat{\bP}}]\right)\right\}.
	\end{aligned}
\end{equation}
Recall that $\mathrm{Hess}f(\hat{\bP})[\bxi_{\hat{\bP}}]$ is an element in $T_{\hat{\bP}}{\widehat{\cM}}$. Thus from \eqref{product_tangent_space}, we can get 
{\setlength\abovedisplayskip{0.1cm}
	\setlength\belowdisplayskip{0.1cm}\begin{equation}
	\begin{aligned}
		\mathrm{Hess}f(\hat{\bP})[\bxi_{\hat{\bP}}]=\left(\mathrm{Hess}f(\hat{\bP}_1)[\bxi_{\hat{\bP}_1}],\cdots,\mathrm{Hess}f(\hat{\bP}_U)[\bxi_{\hat{\bP}_U}]\right),
	\end{aligned}
\end{equation}}
 where $\mathrm{Hess}f(\hat{\bP}_i)[\bxi_{\hat{\bP}_i}]$ can be denoted as
\begin{equation}\label{Hessian_TPC}
	\begin{aligned}
		\mathrm{Hess}f(\hat{\bP}_i)[\bxi_{\hat{\bP}_i}]=&\mathrm{Dgrad}f(\bP_i)[\bxi_{\bP_i}]\\
		& -\mathrm{D}\hat{\lambda}_1[\bxi_{\bP_i}]\bP_i-\hat{\lambda}_1 \bxi_{\bP_i}-\hat{\lambda}_2\bP_i.
	\end{aligned}
\end{equation}
$\mathrm{Dgrad}f(\bP_i)[\bxi_{\bP_i}]$ and $\mathrm{D}\hat{\lambda}_1[\bxi_{\bP_i}]$ remain unknown.
For notational simplicity, let us define
\begin{equation}\label{MFE}
	\begin{aligned}
		\bM_{\ell,i}=\bH_{\ell}\bxi_{\bP_i}\bP_{i}^H\bH_{\ell}^H+\bH_{\ell}\bP_{i}\bxi_{\bP_i}^H\bH_{\ell}^H,\\
		\bF_{\ell,i}=-\bR_{\ell}^{-1}\bM_{\ell,i}\bB_{\ell},\  \bE_{\ell,i}=-\bB_{\ell}\bM_{\ell,i}\bB_{\ell},
	\end{aligned}
\end{equation}
then we can get
{\setlength\abovedisplayskip{0.1cm}
	\setlength\belowdisplayskip{0.1cm}
\begin{equation}\label{DBk}
\begin{aligned}
	\mathrm{D}\bB_{\ell}\left[\bxi_{\bP_i}\right]&=\bF_{\ell,i}+\bF_{\ell,i}^H+\bE_{\ell,i}.
\end{aligned}
\end{equation}}
Then $\mathrm{Dgrad}f(\bP_i)[\bxi_{\bP_i}]$ can be calculated as follows:
{\setlength\abovedisplayskip{0.1cm}
	\setlength\belowdisplayskip{0.1cm}\begin{equation}\label{Dgradf}
	\begin{aligned}
		&\mathrm{Dgrad}f(\bP_i)[\bxi_{\bP_i}]=-2w_i\bH_i^H\bR_i^{-1}\bH_{i}\bxi_{\bP_i}\bC_i\\
		&+2w_i\bH_i^H\bA_i\bC_i\left(\bxi_{\bP_i}^H\bH_{i}^H\bA_{i}+\bA_{i}^H\bH_{i}\bxi_{\bP_i}\right)\bC_i\\
		&+2\sum_{ \ell \neq i}w_{\ell}\bH_{\ell}^H\bB_{\ell}\bH_{\ell}\bxi_{\bP_i}+2\sum_{ \ell \neq i}w_{\ell}\bH_{\ell}^H\mathrm{D}\bB_{\ell}\left[\bxi_{\bP_i}\right]\bH_{\ell}\bP_i.
	\end{aligned}
\end{equation}}
$\mathrm{D}\hat{\lambda}_1[\bxi_{\bP_i}]$ can be calculated as follows:
\begin{equation}\label{Dhatlambda1}
	\begin{aligned}
		&\mathrm{D}\hat{\lambda}_1[\bxi_{\bP_i}]=\frac{1}{P}\Re\sum_{l\neq i}\left\{\bP_{\ell}\Big(\mathrm{Dgrad}f\left(\bP_{\ell}\right)\left[\bxi_{\bP_i}\right]\Big)^H \right\}+\\
		&\frac{1}{P}\Re\left\{\mathrm{tr}\left(\bxi_{\bP_i}^H\mathrm{grad}f(\bP_i)+\bP_i^H\mathrm{Dgrad}f(\bP_i)[\bxi_{\bP_i}]\right)\right\}.
	\end{aligned}
\end{equation}
$\mathrm{Dgrad}f\left(\bP_{\ell}\right)\left[\bxi_{\bP_i}\right]$ for  $\forall \ell\neq i$ remains to be calculated:
{\setlength\abovedisplayskip{0.1cm}
	\setlength\belowdisplayskip{0.1cm}
\begin{equation}
	\begin{aligned}
		&\mathrm{Dgrad}f\left(\bP_{\ell}\right)\left[\bxi_{\bP_i}\right]=2\sum_{j \neq\ell}w_{j}\bH_{j}^H\mathrm{D}\bB_{j}\left[\bxi_{\bP_i}\right]\bH_j\bP_{\ell}\\
		&+2w_{\ell}\bH_{\ell}^H\bR_{\ell}^{-1}\bM_{\ell,i}\bA_{\ell}\bC_{\ell}-2w_{\ell}\bH_{\ell}^H\bA_{\ell}\mathrm{D}\bC_{\ell}\left[\bxi_{\bP_i}\right],
	\end{aligned}
\end{equation}}
where
{\setlength\abovedisplayskip{0.1cm}
	\setlength\belowdisplayskip{0cm}\begin{subequations}\label{DCDB}
	\begin{align}
		\mathrm{D}\bC_{\ell}\left[\bxi_{\bP_i}\right]=&
		-\bC_{\ell}\bA_{\ell}^H\bM_{\ell,i}\bA_{\ell}\bC_{\ell},  \forall  \ell\neq i, \\
		\mathrm{D}\bB_{j}\left[\bxi_{\bP_i}\right]=&
		\begin{cases}
			\bF_{j,i}+\bF_{j,i}^H+\bE_{j,i}, j\neq \ell\&j\neq i,\\
			\bR_i^{-1}\bH_{i}\big(\bxi_{\bP_i}\bC_i\bP_i^H+\bP_i\bC_i\bxi_{\bP_i}^H\big)\times\\
			\bH_{i}^H\bR_i^{-1}+\bE_{i,i}, j=i.
		\end{cases}
	\end{align}
\end{subequations}}
\vspace{0em}

\section{Proof of Riemannian Hessian in \thref{Theo_Riemannian_Gradient_sub_PUPC}}\label{App_Proof_Theo_Riemannian_Hessian_sub_PUPC}
Like \eqref{hessian_first}, we get
\begin{equation}\label{hessian_PUPC_derivation}
	\begin{aligned}
		\mathrm{Hess}f(\tilde{\bP})[\bxi_{\tilde{\bP}}]=\Pi_{T_{\tilde{\bP}}\widetilde{\cM}}^{T_{\bP}\cN}\left(\mathrm{Dgrad}f(\tilde{\bP})[\bxi_{\tilde{\bP}}]  \right),
	\end{aligned}
\end{equation}
which belongs to $T_{\tilde{\bP}}\widetilde{\cM}$, so the following relationship holds:
{\setlength\abovedisplayskip{0.1cm}
	\setlength\belowdisplayskip{0.1cm}\begin{equation}\label{barLambda}
	\begin{aligned}
		\Re\left\{\mathrm{tr}\Bigg(\tilde{\bP}_i^H\Big(\mathrm{Dgrad}f(\tilde{\bP}_i)[\bxi_{\tilde{\bP}_i}]-\tilde{\bP}_i\big[\tilde{\bLambda}_2\big]_i \Big) \Bigg) \right\}=0 \\
		\Rightarrow\big[\tilde{\bLambda}_2\big]_i=\Re\left\{\mathrm{tr}\Bigg(\bP_i^H\left(\mathrm{Dgrad}f(\tilde{\bP}_i)[\bxi_{\tilde{\bP}_i}]\right)\Bigg)\right\}\bI_{d_i}. 
	\end{aligned}
\end{equation}}
From \eqref{product_tangent_space}, we can get 
\begin{equation}
	\begin{aligned}
		\mathrm{Hess}f(\tilde{\bP})[\bxi_{\tilde{\bP}}]=&\left(\mathrm{Hess}f(\tilde{\bP}_1)[\bxi_{\tilde{\bP}_1}],\cdots,\mathrm{Hess}f(\tilde{\bP}_U)[\bxi_{\tilde{\bP}_U}]\right),
	\end{aligned}
\end{equation}
where 
\begin{equation}\label{Hessian_PUPC}
	\begin{aligned}
		&\mathrm{Hess}f(\tilde{\bP}_i)[\bxi_{\tilde{\bP}_i}]=\mathrm{Dgrad}f(\bP_i)[\bxi_{\bP_i}]  \\
		& -\bP_i\mathrm{D}\big[\tilde{\bLambda}_1\big]_i[\bxi_{\bP_i}]- \bxi_{\bP_i}\big[\tilde{\bLambda}_1\big]_i-\bP_i\big[\tilde{\bLambda}_2\big]_i.
	\end{aligned}
\end{equation}
$\mathrm{D}\left[\tilde{\bLambda}_1\right]_i[\bxi_{\bP_i}]$ remains unknown, which can be calculated as follows:
{\setlength\abovedisplayskip{0.1cm}
	\setlength\belowdisplayskip{0.1cm}\begin{equation}\label{DtildeLambda1}
	\begin{aligned}
		&\mathrm{D}\left[\tilde{\bLambda}_1\right]_i[\bxi_{\bP_i}]=\\
		&\frac{1}{P_i}\Re\left\{\mathrm{tr}\left(\bxi_{\bP_i}^H\mathrm{grad}f(\bP_i)+\bP_i^H\mathrm{Dgrad}f(\bP_i)[\bxi_{\bP_i}]\right)\right\}\bI_{d_i},
	\end{aligned}
\end{equation}}
where $\mathrm{Dgrad}f\left(\bP_i\right)\left[\bxi_{\bP_i}\right]$  has been calculated in \appref{App_Proof_Theo_TPC}.
\vspace{0em}

\section{Proof of Riemannian Hessian in \thref{Theo_Riemannian_Gradient_sub_PAPC}}\label{App_Proof_Theo_Riemannian_Hessian_sub_PAPC}
Like \eqref{hessian_first}, we get
{\setlength\abovedisplayskip{0.1cm}
	\setlength\belowdisplayskip{0.1cm}\begin{equation}\label{hessian_PAPC_derivation}
	\begin{aligned}
		\mathrm{Hess}f(\bar{\bP})[\bxi_{\bar{\bP}}]=\Pi_{T_{\bar{\bP}}\overline{\cM}}^{T_{\bP}\cN}\Big(\mathrm{Dgrad}f(\bar{\bP})[\bxi_{\bar{\bP}}]  \Big),
	\end{aligned}
\end{equation}}
which belongs to $T_{\bar{\bP}}\overline{\cM}$, so the following relationship holds:
{\setlength\abovedisplayskip{0.1cm}
	\setlength\belowdisplayskip{0.1cm}\begin{equation}\label{barLambda}
	\begin{aligned}
		\bI_{M_t}\odot\sum_{l=1}^{K}\Re\left\{\bar{\bP}_{\ell}\left(\mathrm{Dgrad}f(\bar{\bP})[\bxi_{\bar{\bP}}]-\bar{\bLambda}_2\bP_{\ell} \right)^H  \right\}=0 \\
		\Rightarrow\bar{\bLambda}_2=\frac{M_t}{P}\bI_{M_t}\odot\sum_{l=1}^{K}\Re\left\{\bP_{\ell}\Big(\mathrm{Dgrad}f(\bar{\bP})[\bxi_{\bar{\bP}}]\Big)^H\right\}.
	\end{aligned}
\end{equation}}
From \eqref{product_tangent_space}, we can get 
\begin{equation}
	\begin{aligned}
		\mathrm{Hess}f(\bar{\bP})[\bxi_{\bar{\bP}}]=\left(\mathrm{Hess}f(\bar{\bP}_1)[\bxi_{\bar{\bP}_1}],\cdots,\mathrm{Hess}f(\bar{\bP}_U)[\bxi_{\bar{\bP}_U}]\right),
	\end{aligned}
\end{equation}
where $\mathrm{Hess}f(\bar{\bP}_i)[\bxi_{\bar{\bP}_i}]$ can be denoted as:
\begin{equation}\label{Hessian_PAPC}
	\begin{aligned}
		&\mathrm{Hess}f(\bar{\bP}_i)[\bxi_{\bar{\bP}_i}]=\\
		&\mathrm{Dgrad}f(\bP_i)[\bxi_{\bP_i}] - \mathrm{D}\bar{\bLambda}_1[\bxi_{\bP_i}]\bP_i-\bar{\bLambda}_1 \bxi_{\bP_i}-\bar{\bLambda}_2\bP_i.
	\end{aligned}
\end{equation}
$\mathrm{Dgrad}f\left(\bP_i\right)\left[\bxi_{\bP_i}\right]$  has been calculated in \appref{App_Proof_Theo_TPC}. $\mathrm{D}\bar{\bLambda}_1[\bxi_{\bP_i}]$ remains unknown, which can be calculated as follows:
{\setlength\abovedisplayskip{0.1cm}
	\setlength\belowdisplayskip{0.1cm}\begin{equation}\label{DbarLambda1}
	\begin{aligned}
		&\mathrm{D}\bar{\bLambda}_1[\bxi_{\bP_i}]=\frac{M_t}{P}\bI_{M_t}\odot\Re\sum_{l\neq i}\left\{\bP_{\ell}\Big(\mathrm{Dgrad}f\left(\bP_{\ell}\right)\left[\bxi_{\bP_i}\right]\Big)^H \right\}+\\
		&\frac{M_t}{P}\bI_{M_t}\odot\Re\left\{\bxi_{\bP_i}\big(\mathrm{grad}f(\bP_i)\big)^H+\bP_i\mathrm{Dgrad}f(\bP_i)[\bxi_{\bP_i}]^H\right\}.
	\end{aligned}
\end{equation}}
$\mathrm{Dgrad}f\left(\bP_{\ell}\right)\left[\bxi_{\bP_i}\right]$ for  $\forall \ell\neq i$ has been calculated in \appref{App_Proof_Theo_TPC}.
\vspace{0em}


\section{Convergence analyses of the proposed Riemannian methods}\label{App_Proof_Convergence}
For the RSD method, $-\mathrm{grad}f\left(\bP^k\right)$ is naturally a descent direction as $g_{\bP^k}\big(-\mathrm{grad}f(\bP^k),$ $\mathrm{grad}f(\bP^k)  \big)< 0$. When the strong Wolfe curvature condition (3.31) in \cite{sato2021riemannian} holds,   the search direction of the RCG method defined in \eqref{UpdateSearchDirection} is a descent direction according to \cite[Lemma 4.1]{sato2021riemannian}. We can reduce the initial step length $\alpha^0$  to ensure the descent property of the search direction. 	Then, let $\left\{\bP^k_{\mathrm{GD}}\right\}$ be an infinite sequence of iterations generated by our proposed RSD or RCG method. $\left\{\bP^k_{\mathrm{GD}}\right\}$ has at least one accumulation point and all the accumulation points of $\left\{\bP^k_{\mathrm{GD}}\right\}$ are critical points as long as $\cL=\left\{ \bP \in\cM \mid f\left(\bP\right)\leq f\left(\bP^0\right) \right\}$ is a compact set from \cite[Theorem 4.3.1]{Absil2009} and \cite[Corollary 4.3.2]{Absil2009}, which holds naturally if the sets formed by TPC, PUPC and PAPC are all compact. From \cite[Theorem 1.4.8]{conway}, a subset of a finite dimensional Euclidean space is compact if and only if it is closed and bounded. $\widehat{\cM}$, $\widetilde{\cM}$ and $\overline{\cM}$ defined in \eqref{sets}  are all subsets of $\cN=\bbC^{M_t\times\sum_{i\in\cU}d_i}$. They are all closed and bounded according to \cite[Definition 1.1.6.]{conway} and \cite[Definition 1.2.15.]{conway}. Therefore, $\widehat{\cM}$, $\widetilde{\cM}$ and $\overline{\cM}$ are all compact manifolds, and the proposed RSD and RCG methods can converge to a critical point.

Let $\left\{\bP^k_{\mathrm{TR}}\right\}$ be an infinite sequence of iterations generated by the proposed RTR method for precoder design. From \cite[Corollary 6.33]{Boumal2020}, $\left\{\bP^k_{\mathrm{TR}}\right\}$ has at least one accumulation point as $\cL$ is compact and all the accumulation points of $\left\{\bP^k_{\mathrm{TR}}\right\}$ are critical points, indicating that the RTR method can converge to a critical point with a superlinear local convergence rate.

More analyses of the convergence of the Riemannian methods with any initialization can be found in \cite{sato2021riemannian} and \cite{boumal2019global}.

\bibliography{Reference}
\bibliographystyle{IEEEtran}

\end{document}